\documentclass[3p,preprint]{elsarticle}
\usepackage[utf8]{inputenc}
\usepackage{amsfonts}
\usepackage{amsmath}
\usepackage{subfigure}
\usepackage{tikz}
\usepackage{enumerate}
\usepackage{ifthen}
\usepackage{hyperref}
\usepackage{todonotes}
\usepackage[linesnumbered,vlined]{algorithm2e}

\SetAlgoLined
\SetKw{KwLet}{let}
\SetKw{KwOr}{or}
\SetKwFor{ElseForEach}{else foreach}{do}{end}
\SetKwFor{WithProb}{with probability}{do}{end}

\usetikzlibrary{patterns,decorations.pathreplacing}

\newcommand{\reals}{\mathbb{R}}

\newcommand{\Oh}{\mathrm{O}}
\newcommand{\oh}{\mathrm{o}}
\newcommand{\Exp}{\mathbf{E}}
\newcommand{\Prob}{\mathbf{Pr}}
\newcommand{\Ccal}{\mathcal{C}}
\newcommand{\Ical}{\mathcal{I}}
\newcommand{\Jcal}{\mathcal{J}}

\newcommand{\Vcal}{\mathcal{V}}
\newcommand{\ALG}{\mathrm{ALG}}
\newcommand{\cost}{\mathrm{cost}}
\newcommand{\opt}{\mathrm{opt}}
\newcommand{\poly}{\mathrm{poly}}
\newcommand{\Ind}{\mathbf{1}}
\newcommand{\hits}{\mbox{ $\mathrm{hits}$ }}

\newtheorem{theorem}{Theorem}
\newtheorem{lemma}[theorem]{Lemma}
\newtheorem{proposition}[theorem]{Proposition}
\newdefinition{definition}[theorem]{Definition}
\newproof{proof}{Proof}

\journal{Theoretical Computer Science}

\begin{document}
\begin{frontmatter}

\title{Query Minimization under Stochastic Uncertainty}
\tnotetext[t1]{Partially supported by Icelandic Research Fund grant 174484-051 and by EPSRC grant EP/S033483/1.
A preliminary version of this paper appeared in volume~12118 of LNCS (LATIN 2020), pp. 181--193, 2020. \url{https://doi.org/10.1007/978-3-030-61792-9_15}.
This work started while M.~S.~de Lima and T.~Tonoyan were at Reykjavik University, during a research visit by S.~Chaplick.
Part of the work was developed while M.~S.~de Lima was at the School of Informatics, University of Leicester, United Kingdom.}

\author[maa]{Steven Chaplick}
\ead{s.chaplick@maastrichtuniversity.nl}

\author[ru]{Magnús M. Halldórsson}
\ead{mmh@ru.is}

\author[uol]{Murilo S. de Lima\corref{cor}}
\ead{mslima@ic.unicamp.br}
\cortext[cor]{Corresponding author}

\author[tech]{Tigran Tonoyan}
\ead{ttonoyan@gmail.com}

\address[maa]{Department of Data Science and Knowledge Engineering, Maastricht University, the Netherlands}

\address[ru]{ICE-TCS, Department of Computer Science, Reykjavik University, Iceland}

\address[uol]{Kópavogur, Iceland}

\address[tech]{Computer Science Department, Technion Institute of Technology, Israel}

\begin{abstract}
We study problems with stochastic uncertainty information on intervals for which the precise value can be queried by paying a cost.
The goal is to devise an adaptive decision tree to find a correct solution to the problem in consideration while minimizing the expected total query cost.
We show that, for the sorting problem, such a decision tree can be found in polynomial time.
For the problem of finding the data item with minimum value, we have some evidence for hardness.
This contradicts intuition, since the minimum problem is easier both in the online setting with adversarial inputs and in the offline verification setting.
However, the stochastic assumption can be leveraged to beat both deterministic and randomized approximation lower bounds for the online setting.
\end{abstract}


\end{frontmatter}

\section{Introduction}

Consider the problem of sorting $n$ data items that are updated concurrently by different processes in a distributed system.
Traditionally, one ensures that the data is strictly consistent, e.g., by assigning a master database that is queried by the other processes, or by running a distributed consensus algorithm.
However, those operations are expensive, and we wonder if we could somehow avoid them.
One different approach has been proposed for the TRAPP distributed database by Olston and Widom~\cite{olston2000queries}, and is outlined as follows.
Every update is sent to the other processes asynchronously, and each process maintains an interval on which each data item may lie.
Whenever the precise value is necessary, a query on the master database can be performed.
Some computations (e.g., sorting) can be performed without knowing the precise value of all data items, so one question that arises is how to perform these while minimizing the total query cost.
Another setting in which this type of problem arises is when market research is required to estimate the data input: a coarser estimation can be performed for a low cost, and more precise information can be obtained by spending more effort in research.
The problem of sorting under such conditions, called the {\bf uncertainty sorting problem with query minimization}, was recently studied by Halldórsson and de Lima~\cite{halldorsson21sortingqueries}.

The study of uncertainty problems with query minimization dates back to the seminal work of Kahan~\cite{kahan91queries} and the
TRAPP distributed database system by Olston and Widom~\cite{olston2000queries}, which dealt with simple problems such as computing the minimum and the sum of numerical data with uncertainty intervals.
These results were later generalized for arbitrary query costs and precision levels by Khanna and Tan~\cite{khanna01queries}.
More recently, more sophisticated problems have been studied in this framework, such as geometric problems~\cite{bruce05uncertainty,charalambous13uncertainty}, network discovery~\cite{beerliova06netdiscovery}, shortest paths~\cite{feder07pathsqueires,welz14thesisqueries}, minimum spanning tree and minimum matroid base~\cite{erlebach14mstverification,erlebach16cheapestset,erlebach08steiner_uncertainty,focke20mstexp,megow17mst,merino19matroidverification,welz14thesisqueries}, linear programming~\cite{ryzhov12lpqueries,yamaguchi18ipqueries}, and NP-hard problems such as the knapsack~\cite{goerigk15knapsackqueries}, scheduling~\cite{arantes18schedulingqueries,durr2020scheduling} and traveling salesman problems~\cite{welz14thesisqueries}.
See~\cite{erlebach15querysurvey} for a survey.

The literature describes two kinds of algorithms for this setting.
Though the nomenclature varies, we adopt the following one.
An {\bf adaptive} algorithm may decide which queries to perform based on results from previous queries.
An {\bf oblivious} algorithm, however, must choose the whole set of queries to perform in advance; i.e., it must choose a set of queries that certainly allow the problem to be solved without any knowledge of the actual values.\footnote{Feder {\em et al.} \cite{feder03medianqueries} call an adaptive algorithm an {\bf online} algorithm, and an oblivious algorithm an {\bf offline} algorithm. We understand that both types of algorithms have to deal with the issue of not knowing the data, and in turn can be considered online exploration problems.}

Two main approaches have been proposed to analyze both types of algorithms.
In the first, an oblivious (adaptive) algorithm is compared to a hypothetical optimal oblivious (adaptive) strategy; this is the approach in \cite{feder07pathsqueires,kahan91queries,merino19matroidverification,olston2000queries}.
However, for more complex problems, and in particular for adaptive algorithms, it usually becomes more difficult to understand the optimal strategy.
A second (more robust) approach is competitive analysis, which is a standardized metric for online optimization~\cite{borodin98online_alg}.
In this setting, both oblivious and adaptive algorithms are compared to an {\bf optimum query set}, which we define next.
For a given realization of the actual values, a set of queries is a {\bf feasible query set} if, after querying all intervals in that set, one can find a solution of the underlying problem; an optimal query set is a feasible query set of minimum cost.
An algorithm (either adaptive or oblivious) is {\bf $\alpha$-query-competitive} if its total query cost is at most $\alpha$ times the cost of an optimum query set.
This type of analysis is performed in~\cite{beerliova06netdiscovery,bruce05uncertainty,erlebach16cheapestset,erlebach08steiner_uncertainty,gupta16queryselection,halldorsson21sortingqueries,kahan91queries,megow17mst,welz14thesisqueries}.
For NP-hard problems, since we do not expect to find the ``correct'' solution in polynomial time, there are two approaches in the literature: either we have an objective function which combines query and solution costs (this is how the scheduling problem is addressed in~\cite{durr2020scheduling}), or we have a fixed query budget and the objective function is based only on the solution cost (as for the knapsack problem in~\cite{goerigk15knapsackqueries}).
Another related problem is that of, given a realization of the precise values, how to compute an optimum query set.
This is often called the {\bf verification}~\cite{beerliova06netdiscovery,charalambous13uncertainty,erlebach14mstverification} or {\bf offline}~\cite{halldorsson21sortingqueries} version of the problem.
This may be interesting both for performing experimental evaluation of the algorithms, as for obtaining insight for solving the corresponding oblivious/adaptive problems (as we discuss in Section~\ref{sec:minimum}).

Competitive analysis is, however, rather pessimistic.
In particular, many problems such as minimum, sorting and spanning tree have a deterministic lower bound of 2 and a randomized lower bound of 1.5 for adaptive algorithms, and a simple 2-competitive deterministic adaptive algorithm, even if queries are allowed to return intervals \cite{erlebach08steiner_uncertainty,gupta16queryselection,halldorsson21sortingqueries,megow17mst}.
For the sorting problem, e.g., Halldórsson and de Lima~\cite{halldorsson21sortingqueries} showed that there is essentially one structure preventing a deterministic adaptive algorithm from having a competitive ratio better than~2.

One natural alternative to competitive analysis is to assume stochastic inputs, i.e., that the precise value in each interval follows a known probability distribution, and we want to build a decision tree specifying a priority  ordering for querying the intervals until the correct solution is found, so that the expected total query cost is minimized.\footnote{Note that, unless some sort of nondeterminism is allowed, the stochastic assumption cannot be used to improve the oblivious results, so we focus on adaptive algorithms.}
In this paper, we study the sorting problem and the problem of identifying the data item with minimum value in this setting.

Some literature is devoted to a similar goal of this paper, but we argue that there are some essential differences.
One first line of work consists of the {\bf stochastic probing problem}~\cite{adamczyk16stochasticprobing,goel06optimizationprobes,guha07optimizationprobes,gupta13stochasticprobing,gupta16stochasticprobing,gupta17stochasticprobingsubmodular,singla18probing}, which is a general stochastic optimization problem with queries.
Even though those works presented results for wide classes of constraints (such as matroid and submodular), they differ in two ways from our work.
First, they assume that only elements that are queried can be in a solution of the underlying problem, or that the objective function is based on the expectation of the non-queried elements.
Second, the objective function is either a combination of the solution and query costs, or there is a fixed budget for performing queries.
Since most of these variants are NP-hard~\cite{goel06optimizationprobes}, some papers~\cite{gupta13stochasticprobing,singla18probing} focused on devising approximation algorithms, while others~\cite{goel06optimizationprobes,gupta16stochasticprobing} on bounding the ratio between an oblivious algorithm and an optimal adaptive algorithm (the {\bf adaptive gap}).
Another very close work is that of Welz~\cite[Section~5.3]{welz14thesisqueries} and Maehara and Yamaguchi~\cite{yamaguchi18ipqueries}, who, like us, assume that a solution may contain non-queried items.
Welz presented results for the minimum spanning tree and traveling salesman problems, but under strong assumptions on the probability distributions.
Maehara and Yamaguchi devised algorithms for a wide class of problems, that also yield improved approximation algorithms for some classical stochastic optimization problems.
However, both works focus on obtaining approximate solutions for the underlying problem, while we wish to obtain an exact one, and they only give asymptotic bounds on the number of queries performed, but do not compare this to the expected cost of an optimum query set.
To sum up, our work gives a better understanding on how the stochastic assumption differs from the competitive analysis, since other assumptions are preserved and we use the same metric to analyze the algorithms: minimizing query cost while finding the correct solution.

\paragraph{Our results}
We prove that, for the sorting problem with stochastic uncertainty, we can construct an adaptive decision tree with minimum expected query cost in polynomial time.
We devise a dynamic programming algorithm which runs in time $\Oh(n^3d^2) = \Oh(n^5)$, where~$n$ is the number of uncertainty intervals and~$d$ is the clique number of the interval graph induced by the uncertainty intervals.
We then discuss why simpler strategies fail, such as greedy algorithms using only local information, or relying on witness sets, which is a standard technique for solving query-minimization problems with adversarial inputs~\cite{bruce05uncertainty,erlebach08steiner_uncertainty}.
We also discuss why we believe that the dynamic programming algorithm cannot be improved to $\oh(n^3)$.

Surprisingly, on the other hand, we present evidence that finding an adaptive decision tree with minimum expected query cost for the problem of finding the data item with minimum value is hard, although the adaptive online version (with adversarial inputs) and the offline (verification) version of the problem are rather simple.
We prove that the (locally) optimal decision tree conditioned to first querying the leftmost interval can be computed easily, and that, in a (globally) optimal decision tree, the leftmost interval is queried first or last.
This also implies that, for any subtree of an optimal decision tree, one branch can be calculated easily.
The hard part of the problem occurs when the global optimum does not query the leftmost interval first, and the question becomes how to find the order in which the other intervals are considered in the ``hard branch'' of the decision tree.
We discuss why various heuristics fail in this case.
A simple approximation result with factor $1 + 1/d_1$ for uniform query costs, where~$d_1$ is the degree of the leftmost interval in the interval graph, follows from the adaptive online version with adversarial inputs~\cite{kahan91queries}.
For arbitrary query costs, we show that the stochastic assumption can be used to beat both deterministic and randomized lower bounds for the adaptive online version with adversarial inputs.

\paragraph{Other related work}
One interesting related problem was studied by van der Hoog {\em et al.}~\cite{vanderhoog19ambiguouspoints}: how to preprocess a set of intervals so that the actual numbers can be sorted efficiently after their precise values are revealed.
Ajtai {\em et al.}~\cite{ajtai16sortingnoise} studied an uncertainty variant of the sorting problem in which a comparison between two values may be imprecise, and the goal is to minimize the number of comparisons to sort the values within a given precision.
Braverman and Mossel~\cite{braverman09sortingnoisy} studied the problem of estimating the most probable ordering of a set of values in two models of uncertainty: (1) when comparisons are not reliable, and (2) when a permutation of the original ordering is sampled within a given probability distribution.

Our work falls into the wide area of stochastic optimization~\cite{birge11stochastic}, and in particular multi-stage stochastic optimization~\cite{pflug14stochastic}.
Problems with uncertainty data described by intervals have also been studied under the framework of robust optimization~\cite{aron04robustuncertainmst,kasperski06intervalregret,yaman01robustmstinterval}.
A classical paper on multi-stage robust optimization is~\cite{chen07robust}.
For surveys on robust optimization, see~\cite{bertsimas11robustsurvey,beyer07robsurvey}.

\paragraph{Organization of the paper}
Section~\ref{sec:sorting} is devoted to the sorting problem with stochastic uncertainty, and Section~\ref{sec:minimum} to the problem of finding the minimum data item.
We conclude the paper with future research questions in Section~\ref{sec:future}.

\section{Sorting}
\label{sec:sorting}

The problem is to sort~$n$ numbers $v_1, \ldots, v_n \in \reals$ whose actual values are unknown.
We are given~$n$ open intervals $I_1, \ldots, I_n$ such that $v_i \in I_i = (\ell_i, r_i)$.
We can query interval~$I_i$ by paying a cost~$w_i$, and after that we know the value of~$v_i$.
We want to find a permutation $\pi : [n] \rightarrow [n]$ such that $v_i \leq v_j$ if $\pi(i) < \pi(j)$ by performing a minimum-cost set of queries.
We focus on adaptive algorithms, i.e., we can make decisions based on previous queries.
We are interested in a stochastic variant of this problem in which~$v_i$ follows some known probability distribution on~$I_i$.
The only constraints are that (1)~values in different intervals have independent probabilities, (2)~for any subinterval $(a, b) \subseteq I_i$, we can calculate $\Prob[v_i \in (a, b)]$ in constant time\footnote{This can be achieved, e.g., by having $\Oh(1)$-time access to the values of the cumulative distribution function (CDF) $F_i$ of the corresponding probability distribution on the $\Oh(n)$ endpoints of all input intervals that are in $I_i$ (we only work with sub-intervals formed by such points). Since such a pre-computation is necessary, stand-alone, and orthogonal to our main topic, we leave it out of the scope of this paper.}, and (3)~the known information on the probability distribution cannot be used to shorten the endpoints of any uncertainty interval\footnote{This can be achieved by assuming that there is a value $0 < \epsilon' < r_i - \ell_i$ such that, for any interval $I_i$ and every $0 < \epsilon \leq \epsilon'$, we have that $\Prob[v_i \in (\ell_i, \ell_i + \epsilon)], \Prob[v_i \in (r_i - \epsilon, r_i)] > 0$. This is natural for continuous probability distributions. Note that this allows regions with probability zero in the middle of the interval. Since we assume that the intervals are open, this assumption can only be satisfied by a discrete probability distribution if the distribution support has infinitely many points; in Section~\ref{sec:probboundary}, we discuss how to modify the algorithm to allow closed intervals, and thus discrete distributions with finite support.}.
The goal is to devise a strategy (i.e., a decision tree) to query the intervals so that the expected query cost is minimized.
More precisely, this decision tree must tell us which interval to query first and, depending on where its value falls, which interval to query second, and so on, until we have enough information to find the permutation $\pi$.

\begin{definition}
\label{def:dep}
Two intervals~$I_i$ and~$I_j$ such that $r_i > \ell_j$ and $r_j > \ell_i$ are {\bf dependent}. Two intervals that are not dependent are {\bf independent}.
\end{definition}

The following lemma and proposition are proved in~\cite{halldorsson21sortingqueries}.
The lemma tells us that we have to remove all dependencies in order to be able to sort the numbers.

\begin{lemma}[\cite{halldorsson21sortingqueries}]
\label{lemma:decideind}
The relative order between two intervals can be decided without querying either of them if and only if they are independent.
\end{lemma}

\begin{proposition}[\cite{halldorsson21sortingqueries}]
\label{fact:offnecessary}
Let~$I_i$ and~$I_j$ be intervals with actual values~$v_i$ and~$v_j$. 
If $v_i \in I_j$ (and, in particular, when $I_i\subseteq I_j$), then~$I_j$ is queried by every feasible query set.
\end{proposition}

Note that the dependency relation defines an interval graph, where we have a vertex for each interval, and two vertices are adjacent if the corresponding intervals intersect~\cite{lekkeikerker62interval}.
Prop.~\ref{fact:offnecessary} implies that we can immediately query any interval containing another interval, hence we may assume a proper interval graph (that is, without nested interval pairs)~\cite{roberts69indifference}.
We may also assume the graph is connected, since the problem is independent for each component, and that there are no single-point intervals, as they would give a non-proper or disconnected graph.

\subsection{An optimal algorithm}
\label{sec:sortingdp}

We describe a dynamic programming algorithm to solve the sorting problem with stochastic uncertainty.
Since we have a proper interval graph, we assume intervals are in the natural total order, with $\ell_1 < \cdots < \ell_n$ and $r_1 < \cdots < r_n$.
We also pre-compute the {\bf regions} $S_1, \ldots, S_t$ defined by the intervals, where $t \leq 2n -1$.
A region is the interval between two consecutive points in the set $\bigcup_{i = 1}^n \{\ell_i, r_i\}$; we assume that the regions are ordered.
We write $S_x = (a_x, b_x)$ with $a_x < b_x$, and we denote by $\Ical_x(y, z) = \{i : S_x \subseteq I_i \subseteq (a_y, b_z) \}$ the indices of the intervals totally contained in $(a_y, b_z)$ that contain $S_x$.
For simplicity we assume that, for any interval~$I_i$ and any region~$S_x$, $\Prob[v_i = a_x] = \Prob[v_i = b_x] = 0$; this is natural for continuous probability distributions, and for discrete distributions we may slightly perturb the distribution support so that this is enforced (we give more detail in Section~\ref{sec:probboundary}).
Since the dependency graph is a connected proper interval graph, we can also assume that each interval contains at least two regions.

Before explaining the recurrence, we first examine how Prop.~\ref{fact:offnecessary} reduces the space of feasible query sets with an example.
In Figure~\ref{fig:sort1}, suppose we first decide to query~$I_3$ and its value falls in region~$S_5$.
Due to Prop.~\ref{fact:offnecessary}, all intervals that contain~$S_5$, namely~$I_2$ and~$I_4$, have to be queried as well.
In Figure~\ref{fig:sort2}, we assume that~$v_2$ falls in~$S_3$ and~$v_4$ falls in~$S_6$.
This forces us to query~$I_1$ but also implies that $I_5$ can be left unqueried.
Therefore, each time we approach a subproblem by first querying an interval~$I_i$ whose value falls in region~$S_x$, we are forced to query all other intervals that contain $S_x$, and so on in a cascading fashion, until we end up with subproblems that are independent of current queried values.
To find the best strategy, we must pick a first interval to query, and then recursively calculate the cost of the best strategy, depending on the region in which its value falls.
Here, the proper interval graph can be leveraged by having the cascading procedure follow the natural order of the intervals.

\begin{figure}[t]
  \centering
  \subfigure[]{\label{fig:sort1}
   \begin{tikzpicture}[thick, scale=0.4]
    \fill[gray!20] (4, 0) -- (5, 0) -- (5, 6) -- (4, 6) -- cycle;    

    \draw (0, 5) node[anchor=east,xshift=0.7mm]{$I_1$} -- (4, 5);
    \draw (1, 4) node[anchor=east,xshift=0.7mm]{$I_2$} -- (5, 4);
    \draw[very thick] (2, 3) node[anchor=east,xshift=0.7mm]{$I_3$} -- (7, 3);
    \draw (3, 2) -- (8, 2) node[anchor=west,xshift=-0.7mm]{$I_4$};
    \draw (6, 1) -- (9, 1) node[anchor=west,xshift=-0.7mm]{$I_5$};
    
    \draw[dotted] (0, 0) -- (0, 6);
    \draw[dotted] \foreach \x in {1, 2, ..., 9} {
      (\x, 0) node[anchor=east,xshift=1.3mm,yshift=-2.5mm]{\ifthenelse{\x=5}{$S_x$}{}} -- (\x, 6) node[anchor=east,xshift=1.3mm,yshift=2.5mm]{$S_\x$}
    };

    \draw [decoration={brace, mirror, raise=0.1cm},decorate,white] (2, -1) -- (6, -1) node [pos=0.5,anchor=north,yshift=-0.2cm] {cascading area};
   
    \fill[black] (4.5, 3) circle (0.15cm);
   \end{tikzpicture}
  }\quad
  \subfigure[]{\label{fig:sort2}
   \begin{tikzpicture}[thick, scale=0.4]
    \fill[gray!20] (2, 0) -- (3, 0) -- (3, 6) -- (2, 6) -- cycle;    
    \fill[gray!20] (5, 0) -- (6, 0) -- (6, 6) -- (5, 6) -- cycle;    

    \draw (0, 5) node[anchor=east,xshift=0.7mm]{$I_1$} -- (4, 5);
    \draw[very thick] (1, 4) node[anchor=east,xshift=0.7mm]{$I_2$} -- (5, 4);
    \draw[dashed] (2, 3) node[anchor=east,xshift=0.7mm]{$I_3$} -- (7, 3);
    \draw[very thick] (3, 2) -- (8, 2) node[anchor=west,xshift=-0.7mm]{$I_4$};
    \draw (6, 1) -- (9, 1) node[anchor=west,xshift=-0.7mm]{$I_5$};
    
    \draw[dotted] (0, 0) -- (0, 6);
    \draw[dotted] \foreach \x in {1, 2, ..., 9} {
      (\x, 0) node[anchor=east,xshift=1.5mm,yshift=-2.5mm]{\ifthenelse{\x=1}{$S_y$}{}\ifthenelse{\x=3}{$S_{z'}$}{}\ifthenelse{\x=5}{$S_x\ $}{}\ifthenelse{\x=6}{$S_{y'}$}{}\ifthenelse{\x=9}{$S_z$}{}} -- (\x, 6) node[anchor=east,xshift=1.3mm,yshift=2.5mm]{$S_\x$}
    };
    
    \draw [decoration={brace, mirror, raise=0.1cm},decorate] (2, -1) -- (6, -1) node [pos=0.5,anchor=north,yshift=-0.2cm] {cascading area};
   
    \fill[black] (4.5, 3) circle (0.15cm);
    \fill[black] (5.5, 2) circle (0.15cm);
    \fill[black] (2.5, 4) circle (0.15cm);
   \end{tikzpicture}
  }
  \caption{A simulation of the querying process for a fixed realization of the values. \subref{fig:sort1}~Querying~$I_3$ first and assuming $v_3 \in S_5$. \subref{fig:sort2}~Assuming $v_2 \in S_3$ and $v_4 \in S_6$.}
  \label{fig:sort}
\end{figure}
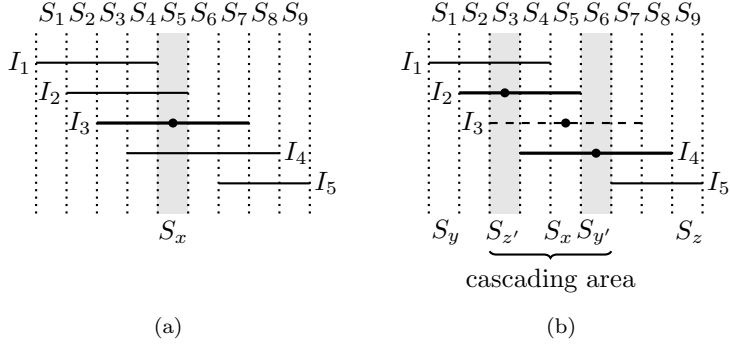

We solve the problem by computing three tables.
The first table, $M$, is indexed by two regions $y, z \in \{1, \ldots, t\}$, and $M[y, z]$ is the minimum expected query cost for the subinstance defined by the intervals totally contained in $(a_y, b_z)$.
Thus, the expected query cost of an optimum decision tree for the whole problem is $M[1, t]$.
To compute $M[y, z]$, we suppose the first interval in $(a_y, b_z)$ that is queried by the optimum decision tree is~$I_i$.
Then, for each region $S_x \subseteq I_i$, when $v_i \in S_x$, we are forced to query every interval~$I_j$ with $j \in \Ical_x(y, z)$ and this cascades, forcing other intervals to be queried depending on where~$v_j$ falls.
So we assume that, for all $j \in \Ical_x(y, z)$, $v_j$~falls in the area defined by regions $z', z' + 1, \ldots, y' -1, y'$, with $z' \leq x \leq y'$, and that this area is minimal (i.e., some point is in~$S_{z'}$, and some point is in~$S_{y'}$).
We call this interval $(a_{z'}, b_{y'})$ the {\bf cascading area} of~$I_i$ in~$\Ical_x(y, z)$.
In Figure~\ref{fig:sort2}, we have $i = 3$, $x = 5$, $z' = 3$ and $y' = 6$.
As the dependency graph is a proper interval graph, the remaining intervals (which do not contain~$S_x$) are split in two independent parts, and we compute the resulting expected query cost using two tables, $L$~and~$R$, which we describe next.
So the recurrence for~$M[y, z]$ is
\begin{displaymath}
\left\{\begin{array}{l}
  0, \qquad \mbox{if } (a_y, b_z) \mbox{ contains less than $2$ intervals; otherwise,} \\
\underbrace{\displaystyle\min_{I_i \subseteq (a_y, b_z)}}_{\begin{array}{c}
  \mbox{\scriptsize first interval}\vspace{-0.1cm}\\
  \mbox{\scriptsize to query}
\end{array}}
\underbrace{\displaystyle\sum_{S_x \subseteq I_i} \Prob[v_i \in S_x]}_{\mbox{\scriptsize where point $v_i$ falls}} \cdot \Bigl(
  \underbrace{\displaystyle\sum_{j \in \Ical_x(y, z)} w_j}_{\mbox{\scriptsize cost of cascading}}
  + \displaystyle\sum_{\substack{z' \leq x \\ y' \geq x}} \underbrace{p(y, z, i, x, z', y')}_{\mbox{\scriptsize cascading area}}
 \cdot \underbrace{\Bigl( \begin{array}{l}
    L[y, z', \min \Ical_x(y, z)] + \\
    + R[y', z, \max \Ical_x(y, z)] 
  \end{array} \Bigr)}_{\mbox{\scriptsize cost of left/right subproblems}}
\Bigr),
 \end{array}\right.
\end{displaymath}
where $p(y, z, i, x, z', y')$ is the probability that $(a_{z'}, b_{y'})$ is the cascading area of~$I_i$ in~$\Ical_x(y, z)$.
We omit how to calculate this probability.

The definitions of~$L$ and~$R$ are symmetric, so we focus on $L$.
For region indices $y,z$, let~$I_{j'}$ be the leftmost interval totally contained in $(a_y, b_z)$.
Now, $L[y, z', j]$ is the minimum expected query cost of solving the subproblem consisting of intervals $I_{j'}, I_{j'+1}, \ldots, I_{j-1}$, assuming that a previously queried point lies in the region~$S_{z'}$.
We ensure that~$z'$ is the leftmost region in $(a_y, b_z)$ that contains a queried point, so that we query all intervals that contain some point.
For example, in Figure~\ref{fig:sort2}, after querying~$I_2$,~$I_3$ and~$I_4$, the left subproblem has $z' = 3$ and $j = 2$.
It holds that~$L$ can be calculated in the following way.
If no interval before~$I_j$ contains~$S_{z'}$, then the cascading is finished and we can refer to table~$M$ for regions $y, y+1, \ldots, z' -1$.
Otherwise $I_{j-1}$ must contain~$S_{z'}$, we query it, and either~$v_{j-1}$ falls to the right of~$a_{z'}$ and we proceed to the next interval, or~$v_{j-1}$ falls in a region~$S_k$ with $k < z'$, and we proceed to the next interval with the leftmost queried point now being  in~$S_k$.
Thus, we have
\begin{displaymath}
 L[y,z',j]= \left\{\begin{array}{ll}
  M[y, z'-1], & \mbox{if } j \leq 1 \mbox{ or } \ell_{j-1} < a_y \mbox{ or } I_{j-1} \not\supseteq S_{z'} \\
  w_{j-1} + \displaystyle\sum_{S_k \subseteq I_{j-1}} \Prob[v_{j-1} \in S_k] \cdot L[y, \min(k, z'), j-1], & \mbox{otherwise.}
 \end{array}\right.
\end{displaymath}
We illustrate this in Figure~\ref{fig:sortleft}.
In Figure~\ref{fig:sortleft1}, the subproblem contains $I_{j-1}, I_{j-2}, \ldots$, and the leftmost queried point is in~$S_{z'}$.
Since $S_{z'} \subseteq I_{j-1}$, we query~$I_{j-1}$ and assume~$v_{j-1}$ falls in a region~$S_k$.
In Figure~\ref{fig:sortleft2}, we have that $k \geq z'$, so we recurse on $L[y, z', j-1]$; this will recurse on $M[y, z'-1]$ in its turn, since $S_{z'} \not\subseteq I_{j-2}$.
In Figure~\ref{fig:sortleft3}, we have that $k < z'$, so we recurse on $L[y, k, j-1]$, which in its turn will have to query~$I_{j-2}$.

\begin{figure}[t]
  \centering
  \subfigure[]{\label{fig:sortleft1}
   \begin{tikzpicture}[thick, scale=0.38]
    \fill[gray!20] (5, 0) -- (6, 0) -- (6, 6) -- (5, 6) -- cycle;    

    \draw (0, 5) -- (3, 5);
    \draw (1, 4) -- (5, 4);
    \draw[very thick] (2, 3) -- (7, 3);
    \draw[dashed] (4, 2) -- (8, 2) node[anchor=west,xshift=-0.7mm]{$I_j$};
    \draw[dashed] (6, 1) -- (9, 1);
    
    \draw[dotted] \foreach \x in {0, 1, ..., 9} {
      (\x, 0) -- (\x, 6) node[anchor=east,xshift=1.3mm,yshift=2mm]{\ifthenelse{\x=1}{$S_y$}{}\ifthenelse{\x=6}{$S_{z'}$}{}}
    };
   
    \fill[black] (5.5, 2) circle (0.15cm);
    \fill[black] (7.5, 1) circle (0.15cm);
   \end{tikzpicture}
  }
  \subfigure[]{\label{fig:sortleft2}
   \begin{tikzpicture}[thick, scale=0.38]
    \fill[gray!20] (5, 0) -- (6, 0) -- (6, 6) -- (5, 6) -- cycle;    

    \draw (0, 5) -- (3, 5);
    \draw (1, 4) -- (5, 4);
    \draw[dashed] (2, 3) -- (7, 3);
    \draw[dashed] (4, 2) -- (8, 2) node[anchor=west,xshift=-0.7mm]{$I_j$};
    \draw[dashed] (6, 1) -- (9, 1);
    
    \draw[dotted] \foreach \x in {0, 1, ..., 9} {
      (\x, 0) -- (\x, 6) node[anchor=east,xshift=1.3mm,yshift=2mm]{\ifthenelse{\x=1}{$S_y$}{}\ifthenelse{\x=6}{$S_{z'}$}{}\ifthenelse{\x=7}{$S_k$}{}}
    };
   
    \fill[black] (6.5, 3) circle (0.15cm);
    \fill[black] (5.5, 2) circle (0.15cm);
    \fill[black] (7.5, 1) circle (0.15cm);
   \end{tikzpicture}
  }
  \subfigure[]{\label{fig:sortleft3}
   \begin{tikzpicture}[thick, scale=0.38]
    \fill[gray!20] (3, 0) -- (4, 0) -- (4, 6) -- (3, 6) -- cycle;    

    \draw (0, 5) -- (3, 5);
    \draw (1, 4) -- (5, 4);
    \draw[dashed] (2, 3) -- (7, 3);
    \draw[dashed] (4, 2) -- (8, 2) node[anchor=west,xshift=-0.7mm]{$I_j$};
    \draw[dashed] (6, 1) -- (9, 1);
    
    \draw[dotted] \foreach \x in {0, 1, ..., 9} {
      (\x, 0) -- (\x, 6) node[anchor=east,xshift=1.3mm,yshift=2mm]{\ifthenelse{\x=1}{$S_y$}{}\ifthenelse{\x=4}{$S_k$}{}\ifthenelse{\x=6}{$S_{z'}$}{}}
    };
   
    \fill[black] (3.5, 3) circle (0.15cm);
    \fill[black] (5.5, 2) circle (0.15cm);
    \fill[black] (7.5, 1) circle (0.15cm);
   \end{tikzpicture}
  }
  \caption{An illustration of the definition of table~$L$. \subref{fig:sortleft1}~$L[y, z', j]$. \subref{fig:sortleft2}~If $k \geq z'$, we recurse on $L[y, z', j-1]$. \subref{fig:sortleft3}~If $k < z'$, we recurse on $L[y, k, j-1]$.}
  \label{fig:sortleft}
\end{figure}
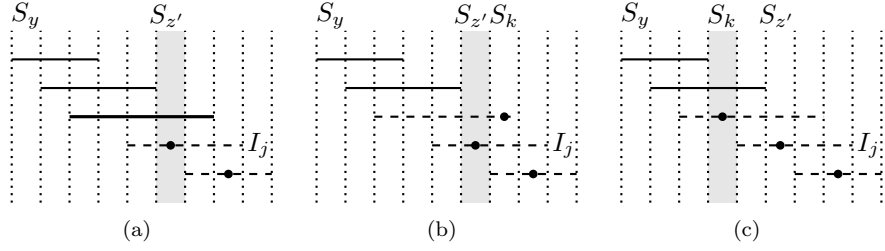

Analogously for table~$R$, let~$I_{j'}$ be the rightmost interval totally contained in $(a_z, b_y)$.
We want to find the best strategy for intervals $I_{j+1}, \ldots, I_{j'}$, assuming that the rightmost queried point is contained in~$S_{y'}$.
Thus, we have
\begin{displaymath}
R[y', z, j] =
\left\{\begin{array}{ll}
 M[y' + 1, z], & \mbox{if } j \geq n \mbox{ or } r_{j+1} > b_z \mbox{ or } I_{j+1} \not\supseteq S_{y'} \\
 w_{j+1} + \displaystyle\sum_{S_k \subseteq I_{j+1}} \Prob[v_{j+1} \in S_k] \cdot R[\max(k, y'), z, j+1], & \mbox{otherwise.}
\end{array}\right.
\end{displaymath}

Observe that the definition of table $L$ only depends on which is the {\em leftmost} region in the cascading area (and, symmetrically, $R$ only depends on the rightmost region).
Therefore, we can simplify the recurrence for table $M$ to
\begin{displaymath}
M[y, z] = \left\{\begin{array}{l}
  0, \hfill \mbox{if } (a_y, b_z) \mbox{ contains less than $2$ intervals} \\
 \displaystyle\min_{I_i \subseteq (a_y, b_z)}
 \displaystyle\sum_{S_x \subseteq I_i} \Prob[v_i \in S_x] \cdot \left(
  \begin{array}{l}
    \displaystyle\sum_{j \in \Ical_x(y, z)} w_j \\
    + \displaystyle\sum_{z' \leq x} p_L(y, z, i, x, z') \cdot L[y, z', \min \Ical_x(y, z)] \\
    + \displaystyle\sum_{y' \geq x} p_R(y, z, i, x, y') \cdot R[y', z, \max \Ical_x(y, z)]
  \end{array}
\right), \hfill \mbox{otherwise.}
 \end{array}\right.
\end{displaymath}
where $p_L(y, z, i, x, z')$ is the probability that $S_{z'}$ is the leftmost region in the cascading area of~$I_i$ in~$\Ical_x(y, z)$, and $p_R(y, z, i, x, y')$ is the probability that $S_{y'}$ is the rightmost region.
We explain how to calculate these probabilities in Section~\ref{sec:probcascading}.

At this point it is not hard to see that the next theorem follows by a standard optimal substructure argument.
We present a proof for the sake of completion.

\begin{theorem}
\label{teo:dp}
The recurrence defined above for $M[1,t]$ correctly defines the minimum expected query cost to solve the stochastic sorting problem with uncertainty.
\end{theorem}

\begin{proof}
We prove that, for any $1 \leq y \leq z \leq t$, the value $M[y, z]$ is the expected query cost of a best decision tree for the subproblem defined by the intervals totally contained in $(a_y, b_z)$.
The proof is by induction on the number of intervals totally contained in $(a_y, b_z)$.
If it contains less than two intervals, then no query has to be done to solve this subproblem and the claim follows, so let us assume it contains at least two intervals.

Let us define more precisely how the decision tree for a subproblem is structured.
Let~$\Ical$ be a collection of intervals and queried points, with at least two dependent elements, and let~$T(\Ical)$ be a best decision tree for solving the subproblem defined by~$\Ical$.
The root of the tree indicates which interval~$I_i$ to query first.
Then, for each region~$S_x$ contained in~$I_i$, the tree has a branch which is the decision tree for the remaining intervals, conditioned to the fact that~$v_i \in S_x$; we can write this subtree as $T((\Ical \setminus I_i) \cup \{v_x\})$, for some $v_x \in S_x$.
Note that, for any $v_x \in S_x$, the cost of $T((\Ical \setminus I_i) \cup \{v_x\})$ is the same, since~$v_x$ will be dependent to the same intervals, and the dependencies between other intervals do not change.
The expected cost of the decision tree encoded by~$T(\Ical)$ is then
\begin{displaymath}
\cost(T(\Ical)) = w_{i} + \sum_{S_x \subseteq I_i} \Prob[v_i \in S_x] \cdot \cost(T((\Ical \setminus I_i) \cup \{v_x\})).
\end{displaymath}
The leaves of the tree will correspond to collections of independent intervals, and will have query cost zero.

If a subtree~$T(\Ical)$ contains a queried value~$v_x$ and a non-queried interval~$I_j$ with $v_x \in I_j$, then Prop.~\ref{fact:offnecessary} says that any feasible query set for~$\Ical$ must query~$I_j$.
This implies that~$I_j$ is queried in the path between the root and any leaf of the tree.
Thus, it is easy to see that there is a decision tree for this subproblem with the same cost in which the first query is~$I_j$.
If more than one interval contains a queried value~$v_x$, then we can query them before other intervals, and in any order, so we can actually query all of them at the same time, and have a root with branches for each combination of regions in which the values fall.

The algorithm starts by querying an interval~$I_i$ and, depending on the region~$S_x$ in which~$v_i$ falls, queries all intervals that contain~$S_x$.
Since we have a proper interval graph, the remaining intervals are divided into two independent suproblems.
Also, if the minimal area containing the regions in which the values fall is the same, then the cost of the subtree is the same, since the same intervals will contain a point queried at this time; this implies that each cascading area is a single disjoint event.
Given a cascading area $(a_{z'}, b_{y'})$, the remaining problem consists of finding the best decision tree for two subproblems: one considering that the intervals to the left of~$S_x$ have not been queried and that the leftmost queried point is in~$S_{z'}$, and another that the intervals to the right of~$S_x$ have not been queried and that the rightmost queried point is in~$S_{y'}$.
This is precisely the definition of tables~$L$ and~$R$; thus, if the recurrences for tables~$L$ and~$R$ are correct, then the theorem follows by an optimal substructure argument.

So let us prove that the recurrence for table~$L[y, z', j]$ is correct; the proof for table~$R$ is analogous.
Let us recall the definition: $L[y, z', j]$ is the minimum expected cost of solving the subinstance of $(a_y, b_z)$ (for some $z \geq z')$ consisting of intervals $I_{j'}, I_{j'+1}, \ldots, I_{j-1}$, where~$I_{j'}$ is the leftmost interval totally contained in $(a_y, b_z)$, assuming that the leftmost queried point is contained in $S_{z'}$.
If $j \leq 1$, then $(a_y, b_{z'-1})$ contains no interval and therefore $M[y, z' -1]$ is zero.
If $\ell_{j-1} < a_y$, then no interval to the left of~$I_j$ is contained in $(a_y, b_z)$, so $(a_y, b_{z'-1})$ contains no interval and $M[y, z' -1]$ is zero.
If $j > 1$ and $\ell_{j-1} \geq a_y$, but $I_{j-1} \not\supseteq S_{z'}$, then all intervals in $I_j, I_{j+1}, \ldots$ have a value to the right of $a_{z'}$, and thus any feasible query set to the intervals totally contained in $(a_y, b_{z' - 1})$ is feasible to complement the current decision tree.
Thus, by an optimal substructure argument, $L[y, z', j] = M[y, z' -1]$.
If $j > 1$, $\ell_{j-1} \geq a_y$, and~$I_{j-1}$ contains~$S_{z'}$, then Prop.~\ref{fact:offnecessary} implies that~$I_{j-1}$ must be queried in any feasible query set of the subproblem, and thus can be the first interval queried in this subproblem.
When querying~$I_{j-1}$, we ensure that the leftmost region with a queried point is updated correctly, so the last term in the recurrence for~$L[y, z', j]$ is correct by an optimal substructure argument.
\qed
\end{proof}

The recurrences can be implemented in a bottom-up fashion in time $\Oh(n^5)$: if we precompute the value of $p_L(y, z, i, x, z')$ and $p_R(y, z, i, x, y')$, then each entry of~$M$ is computed in time $\Oh(n^3)$, and each entry of~$L$ and~$R$ can be computed in linear time.
It is possible to precompute $p_L$ and $p_R$ in time $\Oh(n^3)$, which we discuss in Section~\ref{sec:probcascading}.
A more careful analysis shows that the time consumption of the whole algorithm is $\Oh(n^3 d^2)$, where~$d$ is the clique number of the interval graph.
First, we show that each entry of $M$ can be computed in time $\Oh(n d^2)$.
Note that, in a proper interval graph, an interval contains at most $2d-1$ regions, so we have at most $2d - 1$ choices of~$S_x$ for each~$I_i$.
It holds that $\Ical_x(y, z)$ contains at most~$d$ intervals, since every such interval contains~$S_x$.
Moreover, for a given choice of~$S_x$, note that $z'$ cannot go further left than the leftmost region of the leftmost interval that contains~$S_x$, and since $z' \leq x$ we have at most $2d-1$ choices of $z'$; an analogous argument applies to~$y'$.
Now we argue that we only need to compute $L[y, z', j]$ if $S_{z'} \subseteq I_j$: the leftmost queried point cannot be to the right of $I_j$, since we assume $I_j$ was already queried, and cannot be to the left of $I_j$, since we assume $I_j$ is the leftmost queried interval.
Moreover, it is clear that each entry of~$L$ can be computed in time $\Oh(d)$.
Analogous arguments apply to table~$R$.
Finally, in Section~\ref{sec:probcascading} we also refine the analysis to argue that the probabilities can be precomputed in time $\Oh(n^2 d)$.

We now argue that an actual optimal decision tree can be constructed using polynomial time and space, if we represent it by a directed acyclic graph (DAG).
First, we augment table $M$ with an entry for the interval that is queried first in each subproblem, so we can track what is the best option.
Then, we create one node in the DAG for each entry in tables $M, L, R$.
Each node will then have a polynomial number of children: for table $M$, we have a child for each choice of $x, z', y'$ and, for tables $L$ and $R$, we have a child for each choice of $k$.
Note that the same entry of the table is used by overlapping subproblems, so we can do the same for the nodes in the DAG.
It is not hard to see that this construction can be done using at most as much time as for computing the original tables.

It seems difficult to improve this dynamic programming algorithm to something better than $\Oh(n^3 \cdot \poly(d))$.
Note that the main information that the decision tree encodes is which interval should be queried first in a given independent subproblem (and there are $\mathrm{\Omega}(n^2)$ such subproblems).
We could hope to find an optimal substructure that would not need to test every interval as a first query, and that this information could somehow be inferred from smaller subproblems.
However, consider $I_1 = (0, 100)$, $I_2 = (6, 105)$, and $I_3 = (95, 198)$, with uniform query costs and uniform probability distributions.
The optimum solution for the first two intervals is to first query~$I_1$, but the optimum solution for the whole instance is to start with~$I_2$.
Thus, even though~$I_2$ is a suboptimal first query for the smaller subproblem, it is the optimal first query for the whole instance.
This example could be adapted to a larger instance with more than~$d$ intervals, so that we need at least a linear pass in~$n$ to identify the best first query.

\subsection{Simpler strategies that fail}

It may seem that our dynamic programming strategy above is overly complex, and that a simpler algorithm may suffice to solve the problem.
Below, we show sub-optimality of two  such strategies.

We begin by showing that any greedy strategy that only takes into consideration local information (such as degree in the dependency graph or overlap area) fails.
Consider a $5$-path with intervals $I_1, \ldots, I_5$, where interval $I_i$ has non-empty intersection with intervals $I_{i-1}$ (if $i > 1$) and $I_{i+1}$ (if $i < 5$).
Let each interval have query cost~$1$ and an overlap of~$1/3$ with each of its neighbors, and the exact value be uniformly distributed in each interval.
It can be shown by direct calculation that if we query $I_2$ (or, equivalently,~$I_4$) first, then we have an expected query cost of at most $29/9 = 3.2\bar{2}$, while querying $I_3$ first yields an expected query cost of at least $11/3 = 3.6\bar{6}$.
However, a greedy strategy that only takes into consideration local information cannot distinguish between~$I_2$ and~$I_3$.

One technique that has been frequently applied in the literature of uncertainty problems with query minimization is the use of {\bf witness sets}.
A set of intervals~$W$ is a witness if a correct solution for the underlying problem cannot be computed unless at least one interval in~$W$ is queried, even if all other intervals not in~$W$ are queried.
Witness sets are broadly adopted because they simplify the design of query-competitive adaptive algorithms.
If, at every step, an algorithm queries disjoint witness sets of size at most~$\alpha$, then this algorithm is $\alpha$-query-competitive.
This concept was proposed in~\cite{bruce05uncertainty}.
For the sorting problem, by Lemma~\ref{lemma:decideind}, any pair of dependent intervals constitute a witness set.
However, we cannot take advantage of witness sets for the stochastic version of the problem, even for uniform query costs and uniform probability distributions, and even if we take advantage of the proper interval order.
Consider the following intervals: $(0, 100), (95, 105), (98, 198)$.
The witness set consisting of the first two intervals may lead us to think that either of them is a good choice as the first query.
However, the unique optimum solution first queries the third interval.
(The costs are $843/400 = 2.1075$ if we first query the first interval, $277/125 = 2.216$ if we first query the second interval, and $4182/2000 = 2.0915$ if we first query the third interval.)

\subsection{Computing the probability of a cascading area}
\label{sec:probcascading}

Let us discuss how to calculate $p_L(y, z, i, x, z')$, which is the probability that, given that $v_i \in S_x$, it holds that $v_k > a_{z'}$ for each $k \in \Ical_x(y, z)$, and some $j \in \Ical_x(y, z)$ has $v_j \in S_{z'}$.
(The arguments are symmetric for $p_R(y, z, i, x, y')$, so we omit them.)
Note that, considering the definition of table $M$, we can assume that $\Prob[v_i \in S_x] > 0$; otherwise we don't need to compute $p_L(y, z, i, x, z')$.
We will also assume that $z' \neq x$, and the remaining case can be computed similarly.

Let us fix the values $y, z, x$, and denote $\Ical = \Ical_x(y, z)$, $\Ical_i = \Ical_x(y, z) \setminus \{ i \}$, and $p_L(z', i) = p_L(y, z, i, x, z')$.
Let us further denote by $q_L(z')$ the probability that each $k \in \Ical$ has $v_k > a_{z'}$.
Note that
\begin{displaymath}
 q_L(z')=\prod_{j\in \Ical} \Prob[v_j > a_{z'}],
\end{displaymath}
and $q_L(z')$ can be computed in linear time.
Let
\begin{displaymath}
q_L(z', i) = \frac{q_L(z')}{\Prob[v_i > a_{z'}]}
\end{displaymath}
denote the similar probability defined for the set of intervals $\Ical_i$, instead of $\Ical$, with $\Prob[v_i > a_{z'}] > 0$.

Let $P_L(z', i)$ be the event corresponding to $p_L(z', i)$, and $Q_L(z', i)$ be the event corresponding to $q_L(z', i)$.
Then $P_L(z', i) = Q_L(z', i) \setminus Q_L(z'+1, i)$, and $\Prob[P_L(z', i)] = \Prob[Q_L(z', i)] - \Prob[Q_L(z'+1, i)]$ because $Q_L(z'+1, i) \subseteq Q_L(z', i)$.
Thus, we have that
\begin{displaymath}
p_L(z', i) = q_L(z', i) - q_L(z'+1, i) = \frac{q_L(z')}{\Prob[v_i > a_{z'}]} - \frac{q_L(z'+1)}{\Prob[v_i > a_{z'+1}]}.
\end{displaymath}
Note that all divisions in the equation are defined, because we assumed that $\Prob[v_i \in S_x] > 0$ and $z' < x$.
Since $q_L(z')$ does not depend on $i$, for each fixed $x,y,z$, it can be computed in linear time, so all $q_L(z')$ can be computed for all $x,y,z$ in time $\Oh(n^5)$.
Given this precomputation, all values $p_L(z', i)$ for all $y, z, i, x, z'$ can be precomputed in $\Oh(n^5)$ time.

We can further improve the runtime to $\Oh(n^4)$ as follows.
Here, we fix $z'$ and $y,z$, and compute $q_L(z')$ for various values of $x$, so let us now denote $q_L(x)=q_L(z')$.
Recall that $q_L(x)=\prod_{j\in \Ical_x(y,z)} \Prob[v_j > a_{z'}]$, and note that each factor in this product is non-zero due to assumption~(3) at the beginning of Section~\ref{sec:sorting}.
We compute~$q_L(x)$ sequentially from $x=z'$ to $x = x^*$, where $x^*$ is the rightmost region such that there is an interval that contains both $S_{z'}$ and $S_{x^*}$.
Given $q_L(x)$ for some $x$, $q_L(x+1)$ is computed by removing from the product intervals in $\Ical_x(y,z)\setminus \Ical_{x+1}(y,z)$ and including intervals in $\Ical_{x+1}(y,z)\setminus \Ical_{x}(y,z)$. The remainder of the product is reused. Observe that, during this computation, each interval that contains a region in $(a_{z'}, b_{x^*})$ is included in (and removed from) the product exactly once. Therefore, the computation corresponding to fixed values of $y,z,z'$ can be done in time proportional to the number of regions $x$ between $z'$ and $x^*$ plus the number of intervals that contain some region in $(a_{z'}, b_{x^*})$, that is, in $\Oh(n)$ time. Thus, the computation of the whole table takes $\Oh(n^4)$ time.

Finally, the preprocessing time can be further reduced (still having constant time computation of $p_L(\cdots)$ after preprocessing), by decomposing the table for~$q$ into two parts, based on the following observation. The first table is indexed by $y,x,z'$, while the second one by $z,x,z'$. Consider the value $q_L(y,z,x,z')$. This is a product of probabilities ranging over intervals $\Ical_x(y,z)$. The observation is that $\Ical_x(y,z) = \Ical_x(y,t)\setminus \Jcal_x(z, t)$, where $\Jcal_x(z, t)$ is the set of intervals that contain both $S_x$ and $S_z$ and do not end at $b_z$. Thus,
\begin{displaymath}
q_L(y,z,x,z')= \frac{q_L(y,t,x,z')}{q'_L(z,t,x,z')}
\end{displaymath}
 if the denominator is non-zero, and otherwise $q_L(y,z,x,z')= q_L(y,t,x,z')$, where $q'_L$ is defined similarly as~$q_L$, except it ranges over $\Jcal_x(z, t)$. The subtables~$q_L$ and $q'_L$ can be computed in $\Oh(n^3)$ time using the observations in the previous paragraph.

To refine the analysis, recall that each interval in a proper interval graph contains at most $2d-1$ regions.
Since we only consider pairs $z', x$ such that there is an interval that contains both $S_{z'}$ and $S_x$, we have at most $2d-1$ choices of $x$ for each choice of $z'$.
Thus, we only need time $\Oh(d)$ to compute $q_L(y,z,z')$ for all values of $x$ and fixed $y,z,z'$, and the whole preprocessing can be done in time $\Oh(n^2 d)$.

\subsection{Allowing arbitrary probabilities on interval endpoints}
\label{sec:probboundary}

In this section we discuss how to remove the assumption that $\Prob[v_i = a_x] = \Prob[v_i = b_x] = 0$ for every interval $I_i$ and every region $S_x$.
Remember that we assume that all intervals are open.
Suppose that $\Prob[v_i = b_x] > 0$ for some interval $I_i$ and some region $S_x$ with $x < t$.
(Note that $b_x = a_{x+1}$ if $x < t$, and $\Prob[v_i = a_1] = \Prob[v_i = b_t] = 0$ as all intervals are open.)
We consider three cases.
\begin{enumerate}
  \item If there is some interval $I_j$ with $\ell_j = b_x$ but no interval $j'$ with $r_{j'} = b_x$, then we can simply move the probability $\Prob[v_i = b_x]$ to $\Prob[v_i = b_x - \varepsilon]$, for some $\varepsilon < b_x - a_x$.
  Note that we did not need to query $I_j$ if originally $v_i = b_x$; we could simply make $\pi(i) < \pi(j)$.
  On the other hand, all intervals that originally contained $b_x$ will contain $b_x - \varepsilon$, so they will be queried if originally $v_i = b_x$.
  \item We can do a symmetric operation if there is some interval $I_j$ with $r_j = b_x$ but no interval $j'$ with $\ell_{j'} = b_x$.
  \item If there is an interval $I_j$ with $\ell_j = b_x$ and an interval $I_{j'}$ with $r_{j'} = b_x$, then we can ``shift'' the whole ray $(b_x, +\infty)$ to the right by $2\varepsilon$, for some $\varepsilon > 0$.
  More formally, any interval $I_k$ with $\ell_k \geq b_x$ will have a new left endpoint $\ell'_k = \ell_k + 2\varepsilon$; any interval $I_k$ with $r_k > b_x$ will have a new right endpoint $r'_k = r_k + 2\varepsilon$; for any $\alpha > b_x$, we will move the probability $\Prob[v_k = \alpha]$ to $\Prob[v_k = \alpha + 2\varepsilon]$.
  Then we can move the probability $\Prob[v_i = b_x]$ to $\Prob[v_i = b_x + \varepsilon]$; we did not need to query $I_j$ or $I_{j'}$ if originally $v_i = b_x$ (we could simply make $\pi(j') < \pi(i) < \pi(j)$), and all intervals that originally contained $b_x$ will also contain $b_x + \varepsilon$ (so they will be queried if originally $v_i = b_x$).
\end{enumerate}
Note that the last transformation creates new regions, but the total number of regions will be at most twice the original, so the time consumption does not increase asymptotically.

Moreover, the algorithm can be modified to allow closed and half-open intervals, and thus discrete distributions with finite support.
We only need to be careful when $v_i = b_x$ for some queried interval~$I_i$ and some region $S_x$: we do not need to query all intervals that contain $S_x$, but only intervals $I_j$ with $b_x \in (\ell_j, r_j)$.

\section{Finding the Minimum}
\label{sec:minimum}

We also consider the problem of finding the minimum (or, equivalently, the maximum) of~$n$ unknown values $v_1, \ldots, v_n$.
Note that we do not need to know the precise minimum value, but just the data item that contains it; therefore, in some cases it is not necessary to query the corresponding interval.
We assume that all intervals are open\footnote{The minimum problem has unbounded competitive ratio in the adversarial setting unless we only have open intervals \cite{gupta16queryselection}; we believe that one may obtain better ratios in the stochastic setting, but this particular case did not seem interesting enough for us to devote much attention to it.}, and that they are sorted by the left endpoint, i.e., $\ell_1 \leq \ell_2 \leq \cdots \leq \ell_n$.
Regarding the probability distributions, we also assume constraints (1)--(3) described at the beginning of Section~\ref{sec:sorting}.
Let $\Ical = \{I_1, \ldots, I_n\}$.

We begin by discussing some assumptions we can make.
First, we can assume that the interval graph is a clique: with two independent intervals, we can remove the one on the right.
(However, we cannot assume a proper interval graph, as we did for sorting: the argument that one of the nested intervals has to be queried no longer holds, since here we are only interested in finding the minimum rather than the total order.)
The last assumption is that~$I_1$ does not contain another interval; it is based on the following proposition and implies that~$r_1 = \min_i r_i$.
Note that it also implies that $\ell_1 < \ell_2, \ldots, \ell_n$.

\pagebreak

\begin{proposition}
\label{fact:minI1}
If~$I_1$ contains some~$I_j$, then~$I_1$ is queried in every feasible query set.
\end{proposition}

\begin{proof}
Suppose by contradiction that there is a feasible query set $Q$ that does not contain $I_1$.
Since $\ell_1 = \min_i \ell_i$, we must have $Q = \mathcal{I} \setminus \{I_1\}$ and $v_i 
> r_1$ for every $I_i \in Q$, otherwise we cannot decide if~$v_1$ is the minimum.
However, assuming both $v_j > r_1$ and $I_j \subseteq I_1$ is a contradiction.
\qed
\end{proof}

It is also useful to understand how to find an optimum query set, i.e., to solve the verification problem assuming we know $v_1, \ldots, v_n$.

\begin{lemma}
\label{lemma:minoff}
The optimum query set either
\begin{enumerate}[(a)]
 \item queries interval $I_i$ with minimum $v_i$ and each interval $I_j$ with $\ell_j < v_i$; or
 \item queries all intervals except for $I_1$, if $v_1$ is the minimum, $v_j > r_1$ for all $j > 1$, and this is better than option~(a).
\end{enumerate}
\end{lemma}
Option~(b) can be better not only  due to a particular non-uniform query cost configuration, but also with uniform query costs, when $v_1 \in I_2, \ldots, I_n$.
Note also that~$I_1$ is always queried in option~(a).

\begin{proof}
First let us consider the case when~$v_1$ is the minimum.
If $v_j < r_1$ for some $j > 1$, then~$I_1$ has to be queried even if all other intervals have already been queried, due to Prop.~\ref{fact:minI1}.
Thus, the only situation in which~$I_1$ may not be queried is when $v_j > r_1$ for all $j > 1$, and in this case clearly we have to query all other intervals, since otherwise we cannot decide who is the minimum.

Now we prove that, if~$v_i$ is the minimum with $i \neq 1$, then~$I_i$ must be queried.
Since~$v_i$ is the minimum, all other values fall to the right of $\ell_i$.
In particular, $v_1 \in I_i$, since~$I_1$ has minimum~$r_1$.
Thus, even if all other intervals have already been queried, $I_i$~must be queried due to Prop.~\ref{fact:minI1}.

It remains to prove that, for any $i \geq 1$, if~$v_i$ is minimum and~$I_i$ is queried, then all intervals with $\ell_j < v_i$ must also be queried; we actually prove that $I_1, \ldots, I_j$ must be queried, by induction on~$j$.
The base case is $j = 1$, and if $i \neq 1$ the claim follows from Prop.~\ref{fact:minI1}, since $v_i < r_1$ and $\ell_1 < \ell_i$.
So assume $j > 1$; note that $\ell_{j-1} \leq \ell_j < v_i$ thus, by induction hypothesis, $I_1, \ldots, I_{j-1}$ must be queried.
Since~$v_i$ is the minimum, $v_k \geq v_i > \ell_j$, for $k = 1, \ldots, j-1$.
Therefore, after $I_1, \ldots, I_{j-1}$ are queried, $I_j$~is the leftmost interval, and must be queried due to Prop.~\ref{fact:minI1}, since it contains~$v_i$.
\qed
\end{proof}

We first discuss what happens if the first interval we query is~$I_1$.
In Figure~\ref{fig:minfirst1}, we suppose that $v_1 \in S_3$.
This makes~$I_2$ become the leftmost interval, so it must be queried, since it contains~$v_1$.
At this point we also know that we do not need to query~$I_4$, since $v_1 < \ell_4$.
After querying~$I_2$, we have two possibilities. In Figure~\ref{fig:minfirst2}, we suppose that $v_2 \in S_2$, so we already know that~$v_2$ is the minimum and no other queries are necessary.
In Figure~\ref{fig:minfirst3}, we suppose that $v_2 \in S_6$, so we still need to query~$I_3$ to decide if~$v_1$ or~$v_3$ is the minimum.
Note that, once~$I_1$ has been queried, we do not have to guess which interval to query next, since any interval that becomes the leftmost interval will either contain~$v_1$ or will be to the right of~$v_1$.
Since this is an easy case of the problem, we formalize how to solve it.
The following claim is clear: if we have already queried $I_1, \ldots, I_{i-1}$ and $v_1, \ldots, v_{i-1} > \ell_i$, then we have to query~$I_i$.
(This relies on~$I_i$ having minimum~$\ell_i$ among $I_i, \ldots, I_n$.)
If we decide to first query~$I_1$, then we are discarding option~(b) in the offline solution, so all intervals containing the minimum value must be queried.
The expected query cost is then $\sum_{i = 1}^n w_i \cdot \Prob[I_i \mbox{ must be queried}]$.
Given an interval~$I_i$, it will not need to be queried if there is some~$I_j$ with $v_j < \ell_i$, thus the former probability is the probability that no value lies to the left of~$I_i$.
Since the probability distribution is independent for each interval, the expected query cost will be
\begin{displaymath}
 \sum_{i=1}^n w_i \cdot \prod_{j < i} \Prob[v_j > \ell_i].
\end{displaymath}
This can be computed in $\Oh(n^2)$ time.

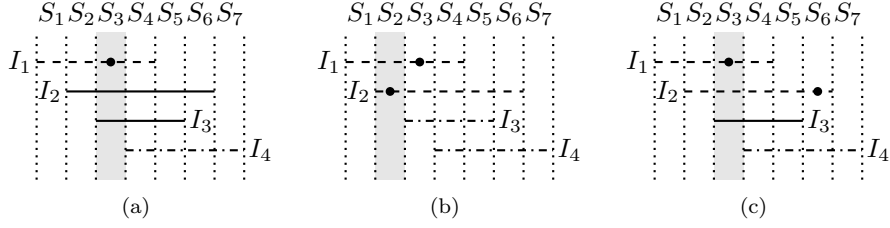
\begin{figure}[t]
  \centering
  \subfigure[]{\label{fig:minfirst1}
   \begin{tikzpicture}[thick, scale=0.39]
    \fill[gray!20] (2, 0) -- (3, 0) -- (3, 5) -- (2, 5) -- cycle;    

    \draw[dashed] (0, 4) node[anchor=east,xshift=0.7mm]{$I_1$} -- (4, 4);
    \draw (1, 3) node[anchor=east,xshift=0.7mm]{$I_2$} -- (6, 3);
    \draw (2, 2) -- (5, 2) node[anchor=west,xshift=-0.7mm]{$I_3$};
    \draw[dashdotted] (3, 1) -- (7, 1) node[anchor=west,xshift=-0.7mm]{$I_4$};

    \draw[dotted] (0, 0) -- (0, 5);
    \draw[dotted] \foreach \x in {1, 2, ..., 7} {
      (\x, 0) -- (\x, 5) node[anchor=east,xshift=1.3mm,yshift=2.5mm]{$S_\x$}
    };
   
    \fill[black] (2.5, 4) circle (0.15cm);
   \end{tikzpicture}
  }
  \subfigure[]{\label{fig:minfirst2}
   \begin{tikzpicture}[thick, scale=0.39]
    \fill[gray!20] (1, 0) -- (2, 0) -- (2, 5) -- (1, 5) -- cycle;    

    \draw[dashed] (0, 4) node[anchor=east,xshift=0.7mm]{$I_1$} -- (4, 4);
    \draw[dashed] (1, 3) node[anchor=east,xshift=0.7mm]{$I_2$} -- (6, 3);
    \draw[dashdotted] (2, 2) -- (5, 2) node[anchor=west,xshift=-0.7mm]{$I_3$};
    \draw[dashdotted] (3, 1) -- (7, 1) node[anchor=west,xshift=-0.7mm]{$I_4$};

    \draw[dotted] (0, 0) -- (0, 5);
    \draw[dotted] \foreach \x in {1, 2, ..., 7} {
      (\x, 0) -- (\x, 5) node[anchor=east,xshift=1.3mm,yshift=2.5mm]{$S_\x$}
    };
   
    \fill[black] (2.5, 4) circle (0.15cm);
    \fill[black] (1.5, 3) circle (0.15cm);
   \end{tikzpicture}
  }
  \subfigure[]{\label{fig:minfirst3}
   \begin{tikzpicture}[thick, scale=0.39]
    \fill[gray!20] (2, 0) -- (3, 0) -- (3, 5) -- (2, 5) -- cycle;    

    \draw[dashed] (0, 4) node[anchor=east,xshift=0.7mm]{$I_1$} -- (4, 4);
    \draw[dashed] (1, 3) node[anchor=east,xshift=0.7mm]{$I_2$} -- (6, 3);
    \draw (2, 2) -- (5, 2) node[anchor=west,xshift=-0.7mm]{$I_3$};
    \draw[dashdotted] (3, 1) -- (7, 1) node[anchor=west,xshift=-0.7mm]{$I_4$};

    \draw[dotted] (0, 0) -- (0, 5);
    \draw[dotted] \foreach \x in {1, 2, ..., 7} {
      (\x, 0) -- (\x, 5) node[anchor=east,xshift=1.3mm,yshift=2.5mm]{$S_\x$}
    };
   
    \fill[black] (2.5, 4) circle (0.15cm);
    \fill[black] (5.5, 3) circle (0.15cm);
   \end{tikzpicture}
  }
  \caption{A simulation of the querying process when we decide to query~$I_1$ first. \subref{fig:minfirst1}~If $v_1 \in S_3$, $I_2$~must be queried, but not~$I_4$. \subref{fig:minfirst2}~If $v_2 \in S_2$, then~$v_2$ is the minimum. \subref{fig:minfirst3}~If $v_2 \in S_6$, then we still have to query~$I_3$.}
  \label{fig:minfirst}
\end{figure}

Now let us consider what happens if an optimum decision tree does not start by querying~$I_1$, but by querying some~$I_k$ with $k > 1$.
When we query~$I_k$, we have two cases: (1)~if~$v_k$ falls in~$I_1$, then we have to query~$I_1$ and proceed as discussed above, querying~$I_2$ if $v_1 > \ell_2$, then querying~$I_3$ if $v_1, v_2 > \ell_3$ and so on; (2) if $v_k \notin I_1$, then~$v_k$ falls to the right of $\ell_i$, for all~$i \neq k$, so essentially the problem consists of finding the optimum decision tree for the remaining intervals, and this value will be independent of $v_k$.
Therefore, the cost of querying $I_k$ first is
\begin{displaymath}
w_k + \Prob[v_k \notin I_1] \cdot \opt(\Ical \setminus \{I_k\}) + \Prob[v_k \in I_1] \cdot \sum_{i \neq k} w_i \cdot \prod_{j < i} \Prob[v_j > \ell_i | v_k \in I_1]. 
\end{displaymath}
Thus, we can see that a decision tree can be specified simply by a permutation of the intervals, since the last term in the last equation is fixed.
More precisely, let $a(1), \ldots, a(n)$ be a permutation of the intervals, where $a(k) = i$ means that~$I_i$ is the $k$-th interval in the permutation.
We have two types of subtrees.
Given a subset~$X_k = \{a(k), \ldots, a(n)\}$ that contains~$1$, let~$\hat{T}_k$ be the tree obtained by first querying~$I_1$, then querying the next leftmost interval in~$X_k$ if it contains~$v_1$ and so on.
The second type of subtree~$T_k$ is defined by a suffix $a(k), \ldots, a(n)$ of the permutation.
If $a(k) \neq 1$, then~$T_k$ is a decision tree with a root querying $I_{a(k)}$ and two branches.
One branch, with probability $\Prob[v_{a(k)} \in I_1]$, consists of~$\hat{T}_{k+1}$; the other branch, with probability $\Prob[v_{a(k)} \notin I_1]$, consists of~$T_{k+1}$.
If $a(k) = 1$, then~$T_k = \hat{T}_k$, unless $k = n$, in which case~$T_n$ will be empty: $I_1$~does not need to be queried, because all other intervals have already been queried and their values fall to the right of~$I_1$.
We have that $\cost(T_k)$
\begin{displaymath}
= \left\{ \begin{array}{ll}
   0, & \mbox{if } a(k) = 1 \mbox{ and } k = n \\
   \cost(\hat{T}_k), & \mbox{if } a(k) = 1 \mbox{ and } k < n \\
   w_{a(k)} + \Prob[v_{a(k)} \in I_1] \cdot \cost(\hat{T}_{k+1} | v_{a(k)} \in I_1)
     + \Prob[v_{a(k)} \notin I_1] \cdot \cost(T_{k+1}), & \mbox{otherwise.}
 \end{array} \right.
\end{displaymath}
Note that in the last case we need to condition~$\cost(\hat{T}_{k+1})$ to the fact that $v_{a(k)} \in I_1$.
We have that
\begin{displaymath}
\cost(\hat{T}_{k}) = 
\sum_{i = k}^n w_{a(i)} \cdot \prod_{\substack{j \geq k \\ a(j) < a(i)}} \Prob[v_{a(j)} > \ell_{a(i)}],
\end{displaymath}
and
\begin{displaymath}
\cost(\hat{T}_{k+1} | v_{a(k)} \in I_1) = 
\sum_{i = k+1}^n w_{a(i)} \cdot \prod_{\substack{j \geq k \\ a(j) < a(i)}} \Prob[v_{a(j)} > \ell_{a(i)} | v_{a(k)} \in I_1].
\end{displaymath}

It holds that, if~$I_1$ is not the last interval in a decision tree permutation, then it is always better to move~$I_1$ one step towards the beginning of the permutation.
This fact is formalized in the following lemma.

\pagebreak

\begin{lemma}
 Given a decision tree permutation $I_k I_1 I_{k'} \cdots$ of a subset~$S$, with $|S| \geq 3$ and $I_1 \in S$, it costs at least as much as the cost of $I_1 I_k I_{k'} \cdots$.
\end{lemma}

\begin{proof}
Let $\Ind[A]$ be the indicator variable of the event $A$, i.e., $\Ind[A] = 1$ if $A$ is true, and zero otherwise.
The cost of $I_k I_1 I_{k'} \cdots $ is
\begin{eqnarray*}
 & & w_k + \Prob[v_k \notin I_1] \cdot \sum_{i \in S \setminus I_k} w_i \cdot \prod_{j < i, j \neq k} \Prob[v_j > \ell_i] \\
  & & + \Prob[v_k \in I_1] \cdot \sum_{i \in S \setminus I_k} w_i \cdot \prod_{j < i} \Prob[v_j > \ell_i | v_k \in I_1] \\
 & = & w_k + \Prob[v_k \notin I_1] \cdot \sum_{i \in S \setminus I_k} w_i \cdot \prod_{j < i, j \neq k} \Prob[v_j > \ell_i] \\
  & & + \Prob[v_k \in I_1] \cdot \sum_{i \in S \setminus I_k} w_i \cdot \max(1 - \Ind[k < i], \Prob[v_k > \ell_i | v_k \in I_1]) \cdot \prod_{j < i, j \neq k} \Prob[v_j > \ell_i] \\
 & = & w_k + \Prob[v_k \notin I_1] \cdot \sum_{i \in S \setminus I_k} w_i \cdot \max(1 - \Ind[k < i], \Prob[v_k > \ell_i | v_k \notin I_1]) \cdot \prod_{j < i, j \neq k} \Prob[v_j > \ell_i] \\
  & & + \Prob[v_k \in I_1] \cdot \sum_{i \in S \setminus I_k} w_i \cdot \max(1 - \Ind[k < i], \Prob[v_k > \ell_i | v_k \in I_1]) \cdot \prod_{j < i, j \neq k} \Prob[v_j > \ell_i] \\
 & = & w_k + \sum_{i \in S \setminus I_k} w_i \cdot ( \Prob[v_k \notin I_1] \cdot \max(1 - \Ind[k < i], \Prob[v_k > \ell_i | v_k \notin I_1]) \\
  & & + \Prob[v_k \in I_1] \cdot \max(1 - \Ind[k < i], \Prob[v_k > \ell_i | v_k \in I_1])) \cdot \prod_{j < i, j \neq k} \Prob[v_j > \ell_i] \\
 & = & w_k + \sum_{i \in S \setminus I_k} w_i \cdot \max(1 - \Ind[k < i], \Prob[v_k > \ell_i]) \cdot \prod_{j < i, j \neq k} \Prob[v_j > \ell_i] \\
 & = & w_k + \sum_{i \in S \setminus I_k} w_i \cdot \prod_{j < i} \Prob[v_j > \ell_i],
\end{eqnarray*}
where the first equality holds since the probability distributions are independent, so $\Prob[v_j > \ell_i | v_k \in I_1] = \Prob[v_j > \ell_i]$ unless $j = k$.
The second equality holds since $\ell_i < r_1$ for all~$i$, so $\Prob[v_k > \ell_1 | v_k \notin I_1] = 1$.
The third equality holds since $\Prob[v_k \notin I_1] + \Prob[v_k \in I_i] = 1$ and by Bayes' theorem.

On the other hand, if we swap~$I_k$ and~$I_1$, then the cost is
\begin{eqnarray*}
\sum_{i \in S} w_i \cdot \prod_{j < i} \Prob[v_j > \ell_i]
 & = & w_k \cdot \prod_{j < k} \Prob[v_j > \ell_k] + \sum_{i \in S \setminus I_k} w_i \cdot \prod_{j < i} \Prob[v_j > \ell_i] \\
 & \leq & w_k + \sum_{i \in S \setminus I_k} w_i \cdot \prod_{j < i} \Prob[v_j > \ell_i],
\end{eqnarray*}
since $\prod_{j < k} \Prob[v_j > \ell_k] \leq 1$, and the last value is precisely the cost of the former permutation.
\qed
\end{proof}

This lemma implies, by induction, that the optimum decision tree either first queries~$I_1$, or has $I_1$ at the end of the permutation.
If~$I_1$ is the last interval in the permutation, then it does not have to be queried if all other values fall to its right.
Thus it may be that, in expectation, having~$I_1$ as the last interval is optimal.

We do not know, however, how to efficiently find the best permutation ending in~$I_1$.
Simply considering which interval begins or ends first, or ordering by $\Prob[v_i \in I_1]$ is not enough.
To see this, consider the following two instances with uniform costs and uniform probabilities.
In the first, $I_1 = (0, 100)$, $I_2 = (5, 305)$ and $I_3 = (6, 220)$; the best permutation is $I_2, I_3, I_1$ and has cost $2.594689$.
If we just extend $I_2$ a bit to the right, making $I_2 = (5, 405)$, then the best permutation is $I_3, I_2, I_1$, whose cost is $2.550467$.

If there was a way to determine the relative order in the best permutation between two intervals $I_j, I_k \neq I_1$, simply by comparing some value not depending on the order of the remaining intervals (for example, by comparing the cost of $I_j I_k I_1 \cdots $ and $I_k I_j I_1 \cdots$), then we could find the best permutation easily.
Unfortunately, the ordering of the permutations is not always consistent, i.e., given a permutation, consider what happens if we swap~$I_j$ and~$I_k$: it is not always best to have~$I_j$ before~$I_k$, or~$I_k$ before~$I_j$.
Consider intervals $I_1 = (0, 1000)$, $I_2 = (3, 94439)$, $I_3 = (8, 6924)$, and $I_4 = (9, 2493)$, with uniform query cost and uniform probability distributions.
The best permutation is $I_4, I_3, I_2, I_1$, and the costs of the permutations ending in~$I_1$ are as follows.
Note that it is sometimes better that~$I_2$ comes before~$I_3$, and sometimes the opposite.
\begin{eqnarray*}
\cost(4, 2, 3, 1) = 3.48611 \quad \cost(2, 4, 3, 1) = 3.48715 \quad \cost(2, 3, 4, 1) = 3.48889 \\
\cost(4, 3, 2, 1) = 3.48593 \quad \cost(3, 4, 2, 1) = 3.48770 \quad \cost(3, 2, 4, 1) = 3.48859
\end{eqnarray*}
This issue also seems to preclude greedy and dynamic programming algorithms from succeeding.
It seems that it is not possible to find an optimal substructure, since the ordering is not always consistent among subproblems and the whole problem.
We have implemented various heuristics and performed experiments on random instances, and could always find instances in which the optimum was missed, even for uniform query costs and uniform probabilities.

Another reason to expect hardness is that the following similar problem is NP-hard~\cite{goel06optimizationprobes}.
Given stochastic uncertainty intervals $I_1, \ldots, I_n$, costs $w_1, \ldots, w_n$, and a query budget~$C$, find a set $S \subseteq \{1, \ldots, n\}$ with $w(S) \leq C$ that minimizes~$\Exp[\min_{i \in S} v_i]$.

Note that the decision version of our problem (whether there is a decision tree with expected cost $\leq \alpha$) is in NP as we can represent a decision tree using linear space (by a permutation of the intervals) and compute its cost in polynomial time.

\subsection{Approximation Algorithms}

Good approximation algorithms have been proposed for the adaptive online version with adversarial inputs~\cite{kahan91queries}.
If query costs are uniform, then first querying~$I_1$ costs at most $\opt + 1$, which yields a factor $1 + 1/d_1$, where~$d_1$ is the degree of~$I_1$ in the interval graph.
For arbitrary costs, there is a randomized $1.5$-approximation algorithm using weighted probabilities in the two strategies stated in Lemma~\ref{lemma:minoff}.
Those results apply to the stochastic version of the problem simply by linearity of expectation.

\begin{theorem}
\label{teo:mindet}
The stochastic minimum problem with uncertainty admits a $(1 + 1/d_1)$-approximation for uniform query costs, and a randomized $1.5$-approximation for arbitrary costs.
\end{theorem}

\begin{proof}
For any path between the root and a leaf in the decision tree, the result holds in the worst case, so it holds in expectation.
\qed
\end{proof}

Those results have matching lower bounds for the adaptive online setting with adversarial inputs, and for arbitrary query costs there is a deterministic lower bound of~$2$.
We show that the stochastic assumption can be used to beat those lower bounds for arbitrary costs.
First, the randomized $1.5$-approximation algorithm can be derandomized, simply by choosing which strategy has smaller expected query cost: either first querying~$I_1$, or first querying all other intervals and if necessary querying~$I_1$.
We know how to calculate both expected query costs in time $\Oh(n^2)$; the latter is
$\sum_{i > 1} w_i + w_1 \cdot \left(1 - \prod_{i > 1} \Prob[v_i > r_1] \right)$.

\begin{theorem}
\label{teo:beatrandmin}
There is a deterministic $1.5$-approximation algorithm for the stochastic minimum problem with uncertainty with arbitrary query costs.
\end{theorem}

\begin{proof}
Let~$\Vcal$ be the set of realizations of the values, and assume~$\Vcal$ is finite.
(Otherwise finiteness can be attained by grouping realizations into equivalence classes based on the partition into regions.)
For each $V \in \Vcal$, let~$C_1(V)$ be the random variable denoting the cost of first querying~$I_1$, then querying~$I_2$ if $v_1 \in I_2$, and so on, and let $w(R) = \sum_{i > 1} w_i$.
We partition~$\Vcal$ in sets~$\Vcal_1$ and~$\Vcal_R$, where $V \in \Vcal_1$ if $\opt(V) = C_1(V)$, and $V \in \Vcal_R$ if $\opt(V) = w(R)$; note that Lemma~\ref{lemma:minoff} guarantees that this is indeed a partition.

Let $\rho = \Prob[V \in \Vcal_R]$ and $\Ccal_1 = \sum_{V \in \Vcal_1} C_1(V) \cdot \Prob[V]$.
If~$\opt^*$ is the expected query cost of the best permutation for the stochastic problem, then
\begin{eqnarray*}
\opt^* & \geq & \sum_{V \in \Vcal} \opt(V) \cdot \Prob[V] \\
 & = & \sum_{V \in \Vcal_1} C_1(V) \cdot \Prob[V] + \sum_{V \in \Vcal_R} w(R) \cdot \Prob[V] \\
 & = & \Ccal_1 + \rho \cdot w(R),
\end{eqnarray*}
where the inequality holds by bounding~$\opt^*$ via a ``fractional'' decision tree.

Now consider the solution for the stochastic problem obtained by the algorithm.
Let~$\ALG_1$ be the cost of first querying~$I_1$, and let~$\ALG_R$ be the cost of first querying all other intervals.
Then
\begin{displaymath}
 \ALG_1 \leq \Ccal_1 + \rho \cdot (w_1 + w(R)) \leq \Ccal_1 + \rho \cdot w(R) + \rho \cdot \frac{\Ccal_1}{1 - \rho} \leq \opt^* + \frac{\rho}{1 - \rho} \cdot \Ccal_1,
\end{displaymath}
where the second inequality holds because $\Prob[V \in \Vcal_1] = 1 - \rho$ and $C_1(V) \geq w_1$ for any~$V$.
On the other hand,
\begin{displaymath}
 \ALG_R \leq (1 - \rho) \cdot (w(R) + w_1) + \rho \cdot w(R) \leq w(R) + \Ccal_1 \leq \opt^* + (1 - \rho) \cdot w(R).
\end{displaymath}
The expected query cost of the algorithm is
\begin{eqnarray*}
\min(\ALG_1, \ALG_R) & \leq & \opt^* + \min \left(\frac{\rho}{1 - \rho} \cdot \Ccal_1 , (1 - \rho) \cdot w(R) \right) \\
 & \leq & \opt^* + \sqrt{\frac{\rho}{1 - \rho} \cdot \Ccal_1 \cdot (1 - \rho) \cdot w(R) } \\
 & = & \opt^* + \sqrt{\Ccal_1 \cdot \rho \cdot w(R)} \\
 & \leq & \opt^* + \frac{\Ccal_1 + \rho \cdot w(R)}{2} \\
 & \leq & \frac{3}{2} \cdot \opt^*,
\end{eqnarray*}
where the second and third inequalities hold by the properties of the geometric mean.
\qed
\end{proof}

Ahead in Theorem~\ref{teo:newrandmin} we propose a more refined algorithm, that we analyze according to the expected approximation ratio.
In the proof above, instead, we bounded the ratio of the expected query cost and the expected optimum query cost (the latter corresponding to the fractional decision tree we use as lower bound).
Since there is no fixed relation between $\Exp[A]/\Exp[B]$ and $\Exp[A/B]$ for random variables $A, B$, we cannot compare the results of Theorems~\ref{teo:beatrandmin} and~\ref{teo:newrandmin}.
To circumvent that, we prove the following proposition.
The example used in the proof was proposed for a similar result in \cite{bampis20meetings}, regarding the ratio of the expected query cost and the expected optimum query cost.

\begin{proposition}
For every $\varepsilon \in (0, 1/2)$, there is an instance for which the algorithm of Theorem~\ref{teo:beatrandmin} has expected approximation ratio $3/2 - \varepsilon$.
\end{proposition}

\begin{proof}
Consider three intervals $I_1, I_2, I_3$, with $\ell_1 < \ell_2 < \ell_3 < r_1 < r_2, r_3$.
Let the query costs be $w_1 = w_3 = k$ and $w_2 = 1$, with $k = (1 - 2\varepsilon)/\varepsilon^2$; note that $k > 0$.
Let the probability distributions be the following:
\begin{center}
\begin{tabular}{llll}
  $\Prob[v_1 \leq \ell_2] = \varepsilon$, & $\Prob[\ell_2 < v_1 \leq \ell_3] = 0$, & $\Prob[v_1 > \ell_3] = 1 - \varepsilon$; &\\
  & $\Prob[v_2 \leq \ell_3] = 1/2$, & $\Prob[\ell_3 < v_2 < r_1] = 0$, & $\Prob[v_2 \geq r_1] = 1/2$;\\
  & & $\Prob[v_3 < r_1] = \varepsilon$, & $\Prob[v_3 \geq r_1] = 1 - \varepsilon$.
\end{tabular}
\end{center}
The algorithm will either query $I_1$ and cascade, with expected query cost
\begin{displaymath}
 k + (1-\varepsilon) \cdot \left(1 + \frac{1}{2} \cdot k \right) = \frac{3k}{2} +1 - \varepsilon \cdot \left( \frac{k}{2} + 1 \right)
\end{displaymath}
(the algorithm queries $I_1$, then $I_2$ if $v_1 > \ell_2$, then $I_3$ if $v_1, v_2 > \ell_3$), or query $\{I_2, I_3\}$, and then $I_1$ if necessary, with expected query cost
\begin{displaymath}
 1 + k + k \cdot \left( 1 - \frac{1}{2} \cdot (1 - \varepsilon) \right) = \frac{3k}{2} + 1 + \varepsilon \cdot \frac{k}{2}.
\end{displaymath}
Since the first option is slightly better, the algorithm will choose that strategy.
Thus, the expected approximation ratio is
\begin{align*}
 & \varepsilon \cdot \frac{k}{k} + (1-\varepsilon) \cdot \left( \frac{1}{2} \cdot \varepsilon \cdot \frac{k+1}{k+1} + \frac{1}{2} \cdot \varepsilon \cdot \frac{2k+1}{2k+1} + \frac{1}{2} \cdot (1 - \varepsilon) \cdot \frac{k+1}{k+1} + \frac{1}{2} \cdot (1 - \varepsilon) \cdot \frac{2k+1}{k+1} \right) = \frac{3}{2} - \varepsilon \ .
\end{align*}
\qed
\end{proof}

Next we describe a more refined deterministic algorithm, which attains expected approximation ratio at most $1.4507$; a pseudocode is given in Algorithm~\ref{alg:rand}.
The idea is, in some sense, to try to decide between first querying~$I_1$ or~$\mathcal{I} \setminus \{I_1\}$ (as the algorithm above), but in a more cautious way, taking into consideration the relationship between the query costs and the probability distributions.
The algorithm has six parameters $\mu_1, \mu_R, \rho_1, \rho_R, \phi_1, \phi_j \geq 1$, which are described in the analysis and are functions of the input.

Let $R = \Ical \setminus \{I_1\}$.
For a set $G \subseteq R$, let $w(G) = \sum_{i \in G} w_i$.
We say that a set $G \subseteq R$ {\bf hits}~$I_1$ if, for some interval $I_i \in G$, it holds that $v_i \in I_1$; note that $\Prob[G \hits I_1] = 1 - \prod_{i \in G} \Prob[v_i \notin I_1]$.

First, we try to find a set $G \subseteq R$ such that $w(G) \leq \alpha \cdot w(R)$ and $\Prob[G \hits I_1] \geq \alpha \cdot \Prob[R \hits I_1]$, for some $\alpha \in [1/4, 3/4]$.
If we manage to find such~$G$, then we use the parameters $\mu_1, \mu_R$, defined later in the analysis, to decide between two strategies.
If $\mu_1 \leq \mu_R$, then we query~$I_1$ and cascade.
Otherwise, we first query $G$, and if it hits~$I_1$ then we query~$I_1$ and cascade.
If~$G$ does not hit~$I_1$, then we query the remaining of~$R$, and query~$I_1$ if necessary.
The high-level reason why it is better to query~$G$ before querying all of~$R$ is the following: if $R$~hits~$I_1$, then there is a good probability (at least $1/4$) that~$G$ also hits~$I_1$, but the cost of~$G$ is bounded by $3/4$ times the cost of~$R$.

Now we describe how to find~$G$.
If some interval $I_j \in R$ has cost at least $3w(R)/4$ and the probability that $R \setminus \{I_j\}$ hits~$I_1$ is at least $\Prob[R \hits I_1] / 4$, then we clearly can take $G = R \setminus \{I_j\}$.
If every interval in~$R$ has cost less than $3w(R)/4$, then there is a set $G' \subseteq R$ with $w(R)/2 \leq w(G') \leq 3w(R)/4$ (it can be computed greedily).
In this case, let $\beta = w(G') / w(R)$; if $G'$~hits~$I_1$ with probability at least $\beta \cdot \Prob[R \hits I_1]$, then we can take $G = G'$ and $\alpha = \beta$.
Otherwise, the complement $R \setminus G'$ has cost $(1-\beta) \cdot w(R)$ and probability of hitting~$I_1$
\begin{displaymath}
 \Prob[R \setminus G' \hits I_1] \geq \Prob[R  \hits I_1] - \Prob[G' \hits I_1] > (1 - \beta) \cdot \Prob[R \hits I_1].
\end{displaymath}
So we let $G=R\setminus G'$ and $\alpha=1-\beta$.
Note that $\beta\in [1/2,3/4]$ and $(1-\beta) \in [1/4, 1/2]$, so $\alpha\in [1/4,3/4]$.

If we do not manage to find the set~$G$, then note that there is an interval~$I_j$ with $w_j \geq 3w(R)/4$, but the probability of $R \setminus \{I_j\}$ hitting~$I_1$ is less than $\Prob[R \hits I_1] / 4$.
We have two strategies, depending on whether $w_1 \leq 3w(R)/4$:
\begin{enumerate}
 \item If $w_1 \leq 3w(R)/4$, then we make a deterministic choice between querying~$I_1$ and cascading, or querying~$R$ and then~$I_1$ if necessary.
That choice is made with the parameters~$\rho_1, \rho_R$, which are defined later in the analysis.
In this case, we show that there is a high chance that~$I_j$ is part of every feasible query set, and it comprises most of the cost of the instance, so the expected query cost is not far from the expected optimum.
 \item If $w_1 > 3w(R)/4$, then we first query $R \setminus \{I_j\}$.
If $R \setminus \{I_j\}$ hits~$I_1$, then we query~$I_1$, and then~$I_j$ if necessary.
Otherwise, we use the parameters~$\phi_1, \phi_j$ to make a deterministic choice between querying~$I_1$ or~$I_j$, then querying the other if necessary.
The idea here is that the cost of $R \setminus \{I_j\}$ is not relevant compared to~$w_1$ and~$w_j$, so we can query $R \setminus \{I_j\}$ immediately, and the main question becomes whether to query~$I_1$ or~$I_j$.
\end{enumerate}

\begin{algorithm}[!ht]
\KwIn{$(I_1, \ldots, I_n, w, \Prob)$}
$G \leftarrow \emptyset$\;
$R \leftarrow \{I_2, \ldots, I_n\}$\;
\eIf{$\exists j \in R : w_j \geq 3 w(R) / 4$}{
  \lIf{$\Prob[R \setminus \{I_j\} \hits I_1] \geq \Prob[R \hits I_1] / 4$}{$G \leftarrow R \setminus \{I_j\}$}
}{
  \KwLet $G' \subseteq R$ such that $w(R) / 2 \leq w(G') \leq 3w(R) / 4$\;
  \KwLet $\beta = w(G') / w(R)$\;
  \lIf{$\Prob[G' \hits I_1] \geq \beta \cdot \Prob[R \hits I_1]$}{$G \leftarrow G'$}
  \lElse{$G \leftarrow R \setminus G'$}
}
\eIf{$G \neq \emptyset$}{
  \eIf{$\mu_1 \leq \mu_R$\label{line:beginrand1}}{
    query $I_1$ and cascade\;
  }{
    query $G$\;
    \lIf{$G$ hits $I_1$}{query $I_1$ and cascade}
    \lElse{query $R \setminus G$, then query $I_1$ if necessary\label{line:endrand1}}
  }
}{
  \eIf{$w_1 \leq 3w(R)/4$}{
    \lIf{$\rho_1 \leq \rho_R$}{query $I_1$ and cascade\label{line:begindet}}
    \lElse{query $R$, then query $I_1$ if necessary\label{line:enddet}}
  }{
    \KwLet $j \in R$ with $w_j \geq 3 w(R) / 4$\;\label{line:beginrand2}
    query $R \setminus \{I_j\}$\;
    \lIf{$R \setminus \{I_j\}$ hits $I_1$}{query $I_1$, then query $I_j$ if necessary}
    \Else{
      \lIf{$\phi_1 \leq \phi_j$}{query $I_1$, then query $I_j$ if necessary}
      \lElse{query $I_j$, then query $I_1$ if necessary\label{line:endrand2}}
    }
  }
}
\vspace{0.2cm}
\caption{\label{alg:rand} Refined approximation algorithm for the stochastic minimum problem with uncertainty.}
\end{algorithm}

We now analyze the expected approximation factor of this algorithm.

\begin{theorem}
\label{teo:newrandmin}
Algorithm~\ref{alg:rand} has expected approximation factor at most $1.4507$ for the stochastic minimum problem with uncertainty with arbitrary query costs.
\end{theorem}
\begin{proof}
Throughout the proof, we denote by $C_1$ the random variable consisting of the cost of querying~$I_1$ and cascading.
Let~$k$ be the maximum index such that $w_1 + \ldots w_k \leq w(R)$, and let $p_1 = \Prob[v_1 > \ell_{k+1}] = \Prob[v_1 \in I_{k+1}]$ (cf.\ Prop.~\ref{fact:minI1}).
(Note that $k < n$ because $w_1 > 0$.)
Note that $p_1 = \Prob[C_1 > w(R) | R \mbox{ does not hit } I_1]$: this is because, if we assume that all the values in~$R$ are to the right of~$r_1$, then every interval~$I_i \in R$ with $v_1 \in I_i$ must be queried.
We also point out that~$p_1$ is independent from the choices of the values of the intervals in~$R$.
To simplify notation, let $p_R = \Prob[R \hits I_1]$.

We divide the analysis in three cases: (1) when $G \neq \emptyset$, in which case the algorithm runs Lines~\ref{line:beginrand1}--\ref{line:endrand1}, (2) when $G = \emptyset$ and $w_1 \leq 3w(R)/4$, corresponding to Lines~\ref{line:begindet}--\ref{line:enddet}, and (3) when $G = \emptyset$ but $w_1 > 3w(R)/4$, in which case we run Lines~\ref{line:beginrand2}--\ref{line:endrand2}.

\emph{Case 1}, $G \neq \emptyset$.
We have $w(G) \leq \alpha \cdot w(R)$ and $\Prob[G \hits I_1] \geq \alpha \cdot \Prob[R \hits I_1]$, with $\alpha \in [1/4, 3/4]$.
We divide the analysis in two cases, depending on the behavior of the algorithm.

First, consider the case that the algorithm queries~$I_1$ and cascades.
If~$R$ hits $I_1$, which happens with probability~$p_R$, then the optimal strategy is to query $I_1$ and cascade, as discussed in Section~\ref{sec:minimum}, so we are optimal.
Otherwise, the optimum cost is $\min(w(R), C_1)$, while the algorithm pays $C_1$.
In this case, the approximation ratio is 1 if $C_1\le w(R)$ (which happens with probability $1-p_1$); otherwise (with probability~$p_1$), the ratio is at most $(w_1 + w(R))/w(R) = 1 + w_1/w(R)$.
Thus, the expected approximation ratio is at most
\begin{equation}\label{E:case1e1}
p_R\cdot 1+(1-p_R) \cdot \left(1-p_1 + p_1 \cdot \left(1+\frac{w_1}{w(R)}\right)\right) = 1 + (1-p_R) \cdot p_1\cdot \frac{w_1}{w(R)} =: \mu_1\ .
\end{equation}

Next, suppose we first query~$G$.
Recall that, with probability~$p_R$, $R$ hits $I_1$ and it is optimal to query $I_1$ and cascade.
Let $\gamma = \Prob[G \hits I_1 | R  \hits I_1]$.
If $G$~hits~$I_1$, then an upper bound on the query cost of the algorithm is $w(G) + C_1$, while the optimum is $C_1$. Since $C_1\ge w_1$, the approximation ratio in this case is $(C_1 + w(G))/C_1\le 1+w(G)/w_1\le 1+\alpha w(R)/w_1$.
Otherwise, in the worst case the algorithm pays $w(R) + w_1$, while again, the optimum is at least $w_1$, so the approximation ratio is at most $(w_1+w(R))/w_1=1+w(R)/w_1$.
Thus, the expected approximation ratio conditioned to $R$ hitting~$I_1$ is at most
\begin{align*}
   & \gamma \cdot \left( 1 + \frac{\alpha \cdot w(R)}{w_1} \right) + (1 - \gamma) \cdot \left(1 + \frac{w(R)}{w_1} \right) = 1+(1-\gamma +\gamma\alpha)\cdot \frac{w(R)}{w_1}\\
  &\le  1+(1-\alpha +\alpha^2)\cdot \frac{w(R)}{w_1}\le 1+\frac{13w(R)}{16w_1}\ ,
\end{align*}
where the first inequality holds because $1\ge \gamma\ge \alpha$ (hence $\gamma(\alpha-1)\le \alpha(\alpha-1)$), and the last inequality holds because the maximum for $\alpha^2 + 1-\alpha$ in the interval $[1/4, 3/4]$ is $13/16$.

If $R$ does not hit $I_1$, which happens with probability $1 - p_R$, the  optimum cost is $\min(C_1, w(R))$, while the algorithm pays $w(R)$. Conditioned to this event, the expected approximation ratio is at most $p_1\cdot 1 + (1-p_1)\cdot w(R)/w_1$, by conditioning whether $C_1\ge w(R)$ or not.

Altogether, the expected approximation ratio when the algorithm first queries~$G$ is at most
\begin{align}
&p_R\cdot \left(1 + \frac{13 w(R)}{16w_1}\right) + (1-p_R)\cdot \left(p_1 + (1-p_1)\cdot \frac{w(R)}{w_1}\right)\notag\\
&=1+\frac{13}{16} \cdot z + (1-p_R) \cdot \left(p_1 \cdot \left(1-z\right)+\frac{3}{16} \cdot z - 1\right)=: \mu_R\ ,\label{E:case1e2}
\end{align}
where we denote $z=w(R)/w_1$.

By definition, the expected approximation ratio of the algorithm is at most $\min(\mu_1, \mu_R)$.
Our goal now is to show that this minimum is upper-bounded by the claimed approximation factor.

First, consider the case when $0 < z \le 1$; then, both $\mu_1$ and $\mu_R$ are increasing in $p_1$, so we may let $p_1=1$.
Then the functions to bound are $f_1(z,p_R)=1+(1-p_R)/z$ and $f_2(z,p_R)=1+(13/16) \cdot p_R \cdot z$.
If $p_R\in \{0,1\}$, then $\min(f_1,f_2)=1$, so assume $p_R\in (0,1)$.
In this case, for every fixed $p_R$, it holds that $f_1$ is a decreasing function of $z$, while $f_2$ is an increasing function of $z$.
Thus, for every $p_R$, the $z$ that maximizes $\min(f_1,f_2)$ is given by $f_1(z,p_R)=f_2(z,p_R)$, that is, $z=\sqrt{\frac{16(1-p_R)}{13p_R}}$.
Plugging that value of $z$ back into $f_2$, we obtain
$
1+\sqrt{\frac{13}{16}(p_R-{p_R}^2)},
$
which is maximized at $p_R=1/2$, giving us the bound $1+\sqrt{13}/8 < 1.4507$.

Now assume $z>1$. Note that we may assume $z< 16/3$, as otherwise $\mu_1\le 1+3/16$.
Under such assumption, the coefficient of $1-p_R$ in~(\ref{E:case1e2}) is negative, which means that for every fixed $p_1$ and $z$, $\mu_1$ is a decreasing function of $p_R$, while $\mu_R$ is an increasing function of $p_R$. Thus, for each $z$ and $p_1$, $\min(\mu_1,\mu_R)$ is maximized when $\mu_1=\mu_R$. Solving the last equation for $p_R$ (ignoring the restriction $p_R\in (0,1)$), we get 
\[
1-p_R=\frac{13z/16}{p_1(z-1+1/z)+1-3z/16}\ .
\]
The common value of $\mu_1$ and $\mu_R$ for such $p_R$, plugging the expression above in~(\ref{E:case1e1}), is then
\[
1+\frac{13/16}{z-1+1/z+(1-3z/16)/p_1}\le 1+\frac{13/16}{13z/16 + 1/z}\le 1+\sqrt{13}/8\ ,
\]
where the first inequality holds by letting $p_1 = 1$, which we can do since $z < 16/3$ implies that the left hand side is an increasing function of $p_1$.
The second inequality is obtained by minimizing the denominator by setting its derivative to 0: $13/16 -1/z^2=0$, $z=\sqrt{16/13}$.
Concluding, we see that the expected approximation ratio in Case 1 is always bounded by $1+\sqrt{13}/8< 1.4507$.

\emph{Case 2}, $G = \emptyset$, $w_1 \leq 3w(R)/4$.
In this case we make a deterministic choice between two options: either query~$I_1$ and cascade, or query~$R$ and then query~$I_1$ only when~$R$ hits~$I_1$.
The choice is based on estimates $\rho_1, \rho_R$ of the expected approximation ratio we get in either option.

First, consider the case when we query~$I_1$ and cascade.
With probability~$p_R$, the optimum is to query~$I_1$ and cascade as discussed in Section~\ref{sec:minimum}.
Assuming~$R$ does not hit~$I_1$, with probability\ $1-p_1$ again the optimum is to query~$I_1$ and cascade.
In the remaining case, we have approximation ratio at most $(w_1+w(R))/w(R)$, as the optimum only queries~$R$. Putting these cases together, we get approximation ratio at most
\begin{equation}\label{E:mincomp1}
p_R \cdot 1 + (1 - p_R) \cdot \left[ (1 - p_1) \cdot 1 + p_1 \cdot \frac{w_1+w(R)}{w(R)} \right] = 1 + p_1(1-p_R) \cdot \frac{w_1}{w(R)} =: \rho_1.
\end{equation}

Next, let us estimate the approximation ratio when we choose to first query~$R$.
We divide in several cases to organize the argument.
\begin{enumerate}
  \item With probability $1-p_R$, $R$~does not hit~$I_1$.
  Conditioned on this event, with probability~$p_1$, the optimum queries~$R$ only, otherwise the optimum costs $C_1 \geq w_1$.
  Thus, we have approximation ratio at most
  \begin{displaymath}
    p_1+(1-p_1) \cdot \frac{w(R)}{w_1}.
  \end{displaymath}
  
  \item With probability $p_R$, $R$~hits~$I_1$, so the optimum is to query~$I_1$ and cascade, and the algorithm will query all the intervals.
  Remember that $R$~contains an element~$I_j$ with $w_j \geq 3 w(R)/4$, but
  \begin{displaymath}
   \Prob[R \setminus \{I_j\} \hits I_1] < \Prob[R \hits I_1] /4.
  \end{displaymath}
  Thus
  \begin{eqnarray}
  \Prob[(v_j \in I_1) \land (R \setminus \{I_j\} \mbox{ does not hit }I_1)] & = & \Prob[R \hits I_1] - \Prob[R \setminus \{I_j\} \hits I_1] \notag\\
   & > & 3 \cdot \Prob[R \hits I_1] / 4.\label{eq:probj}
  \end{eqnarray}
  Let $p'_1 = \Prob[v_1 \in I_j]$.
  Again, note that~$p'_1$ is independent from the choices in~$R$.
  Conditioned on the event when $R$~hits~$I_1$, we have the following cases.
  \begin{enumerate}
    \item With probability~$p'_1$, we have that $v_1 \in I_j$.
    Conditioned to this event, we have two subcases.
    \begin{enumerate}
      \item Since we are conditioned to $R$~hitting~$I_1$, by Equation~(\ref{eq:probj}) we have probability at least $3/4$ that $R \setminus \{I_j\}$ does not hit $I_1$ but $v_j \in I_1$.
      Therefore, querying $I_j$ must be part of any feasible query set, since $v_1 \in I_j$ and every interval $I_{j'} \in R \setminus \{I_j\}$ has $v_{j'} > r_1$ (including those with $j' < j$).
      So we have approximation ratio at most
      \begin{displaymath}
        \frac{w_1 + w(R)}{w_1 + w_j} \leq \frac{w_1 + w(R)}{w_1 + 3w(R)/4}.
      \end{displaymath}
      \item With probability at most $1/4$, we use the bound of $(w_1 + w(R))/w_1$.
    \end{enumerate}
    Since $\frac{1}{w_1 + 3w(R)/4} \leq \frac{1}{w_1}$, we have that the ratio conditioned to $v_1 \in I_j$ is at most
    \begin{displaymath}
     (w_1 + w(R)) \cdot \left(\frac{3}{4} \cdot \frac{1}{w_1 + 3w(R)/4} + \frac{1}{4} \cdot \frac{1}{w_1} \right).
    \end{displaymath}
    
    \item With probability $1 - p'_1$, we simply bound the ratio by $(w_1 + w(R))/w_1$.
  \end{enumerate}
  Thus we have that the ratio conditioned to $R$~hitting~$I_1$ is at most
  \begin{displaymath}
    (w_1 + w(R)) \cdot \left[ p'_1 \cdot \left( \frac{3}{4w_1 + 3w(R)} + \frac{1}{4w_1} \right) + \frac{1 - p'_1}{w_1} \right].
  \end{displaymath}
  We then claim that $p'_1 \geq p_1$.
  Suppose by contradiction that $p'_1 < p_1$, then we have that $j > k + 1$.
  By the definition of $k$, it holds that
  \begin{eqnarray*}
  w_1 + \ldots + w_{k+1} & > & w_2 + \ldots + w_n \\
  w_1 & > & w_{k+2} + \ldots + w_j + \ldots + w_n \geq w_j \geq 3w(R)/4,
  \end{eqnarray*}
  a contradiction since we assume that $w_1 \leq 3w(R)/4$.
  Thus, since $\frac{1}{w_1 + 3w(R)/4} \leq \frac{1}{w_1}$, the approximation ratio conditioned to $R$~hitting~$I_1$ is at most
  \begin{displaymath}
    (w_1 + w(R)) \cdot \left[ p_1 \cdot \left( \frac{3}{4w_1 + 3w(R)} + \frac{1}{4w_1} \right) + \frac{1 - p_1}{w_1} \right].
  \end{displaymath}
\end{enumerate}
Summing up, the expected approximation ratio assuming that we first query~$R$ is at most
\begin{eqnarray}
 & & (1-p_R) \cdot \left[p_1 + (1-p_1) \cdot \frac{w(R)}{w_1}\right] + p_R \cdot (w_1 + w(R)) \cdot \left[ p_1 \cdot \left( \frac{3}{4w_1 + 3w(R)} + \frac{1}{4w_1} \right) + \frac{1 - p_1}{w_1} \right] \notag\\
 & = & (1-p_R) \cdot p_1 + p_R +\left(1-p_1+\frac{p_R \cdot p_1}{4}\right) \cdot \frac{w(R)}{w_1} +  \frac{3p_R \cdot p_1}{4} \cdot \frac{w(R)}{4w_1+3 w(R)} \notag\\
 & = & (1-p_R) \cdot p_1 + p_R +\left(1-p_1+\frac{p_R \cdot p_1}{4}\right)z +  \frac{3p_R \cdot p_1}{4} \cdot \frac{z}{3z+4} =: \rho_R, \label{E:mincomp4}
\end{eqnarray}
where $z=w(R)/w_1\ge 4/3$.

Thus, our goal is to bound from above the minimum of the last expression and (\ref{E:mincomp1}). 
Consider their difference:
\[
(1 - p_R - z)(1-p_1) + p_1 \left(\frac{1-p_R}{z} - \frac{p_R z}{4} - \frac{3p_Rz}{4(3z+4)}\right).
\]
Note that $1-p_R-z\le 0$, so the difference is non-positive whenever  the last three terms have a non-positive sum, that is, $p_R \ge \frac{4}{z^2+3z^2/(3z+4)+4}$. That is, for all such values of $p_R$, (\ref{E:mincomp1}) is the smaller one. It is then maximized when $p_1=1$ and $p_R=\frac{4}{z^2+3z^2/(3z+4)+4}$, in which case it equals
\[
1+\frac{z+3z/(3z+4)}{z^2+3z^2/(3z+4)+4}.
\]
Now, it can be shown using standard univariate optimization tools that this function is always bounded by~$1.289$ in $[4/3,\infty)$.

Thus, we can assume that $p_R < \frac{4}{z^2+3z^2/(3z+4)+4}$. Note that (\ref{E:mincomp1}) is an increasing function of $p_1$, while (\ref{E:mincomp4}) is a non-increasing function of $p_1$, since the coefficient
\[
1-p_R -z + \frac{p_Rz}{4} +\frac{3p_Rz}{12z+16} = (1-p_R)(1-z)- p_R \cdot \frac{3z-1}{4} - \frac{p_R}{3z+4}
\]
of $p_1$ is non-positive, for all $z>1$ and $p_R\in [0,1]$. Hence, the minimum of (\ref{E:mincomp1}) and (\ref{E:mincomp4}) is maximized
when they are equal. Thus, we equate them and solve for $p_1$, to find
\[
p_1=\frac{p_R+z-1}{p_R+z-1+(1-p_R)/z - p_Rz/4 - 3p_Rz/(12z+16)}.
\]
Note that $p_1<1$, due to our assumption on $p_R$. Plugging back into (\ref{E:mincomp1}), we get 
\begin{eqnarray}
\notag & & 1+\frac{(1-p_R)(p_R+z-1)}{p_Rz+z^2-z + 1-p_R-p_Rz^2/4-3p_Rz^2/(12z+16)}\\
& = & 1+\frac{(1-p_R)(p_R+z-1)}{A-Bp_R}\label{E:mincomp31}\\
& = & 1+(p_R+z-1)\left(\frac{1}{B} + \frac{1-A/B}{A-Bp_R}\right),\label{E:mincomp32}
\end{eqnarray}
where $A=z^2-z+1\ge 1$, and $B=1+z^2/4+3z^2/(12z+16)-z>0$ (for all $z\ge 1$). Note also that $A>B$ holds for all $z\ge 1$. We can rule out $z>11/5$, as follows. Assuming $z>11/5$, our bound on $p_R$ implies $p_R < 0.392$. Let us set $p_R=0.392$ in the first parentheses in (\ref{E:mincomp32}). Then we obtain a decreasing function of $p_R$, since $1-A/B<0$, which is maximized for $p_R=0$, giving us the following function of $z$: $1+\frac{z-0.608}{z^2-z+1}$. The latter is easily checked to be less than 1.44 for all $z > 11/5$.

Thus, we assume that $z \in [4/3, 11/5]$. It can be checked that the value of $B$ in this interval is bounded by $0.353$. We replace $B=0.353$ and $p_R=1$ in the denominator of (\ref{E:mincomp31}), to get the function
\[
1+\frac{(1-p_R)(z-(1-p_R))}{A-0.353}.
\]
We optimize this for $p_R$. If $z\le 2$, the optimum is $p_R=1-z/2$, giving us the function $1+\frac{z^2/4}{A-0.353}$, while for $z> 2$, the optimal value is $p_R=0$, giving us the function $1+\frac{z-1}{A-0.353}$. Both functions can rather easily be checked to be bounded by $1.41$, for $z \in [4/3, 11/5]$.

\emph{Case 3}, $G = \emptyset$, $w_1 > 3w(R)/4$.
To simplify notation, let us write $R' = R \setminus \{I_j\}$.
Remember that $w_j \geq 3w(R)/4$; this implies $w(R') \leq w(R)/4$, so $w(R') \leq w_j/3$ and $w(R') < w_1/3$.

First let us consider the case that $R'$ hits $I_1$, so the optimum strategy is to query~$I_1$ and cascade.
If~$I_j$ is not queried, then the approximation ratio is at most
\begin{displaymath}
 \frac{w(R') + w_1}{w_1} \leq \frac{4}{3}.
\end{displaymath}
If~$I_j$ is queried in the cascading, then it must be part of every feasible query set, so the approximation ratio is at most
\begin{displaymath}
 \frac{w(R') + w_1 + w_j}{w_1 + w_j} \leq \frac{7}{6}.
\end{displaymath}

Now let us assume that $R'$ does not hit $I_1$.
Let $p'_1 = \Prob[v_1 \in I_j]$ and $p_j = \Prob[v_j \in I_1]$; note that those are independent events.

First, consider the situation in which the algorithm queries~$I_1$ before~$I_j$.
With probability~$p'_1$, the algorithm also queries~$I_j$; if $v_j \in I_1$ (which happens with probability~$p_j$), then both~$I_1$ and~$I_j$ need to be queried to decide which is the minimum, otherwise the optimum strategy is to query~$R$ only, as $C_1\geq w_1 + w_j$ but $w(R) = w_j + w(R') \leq w_j + w_1/3$.
With probability $1 - p'_1$, the algorithm does not query~$I_j$, and the optimum strategy is to query~$I_1$ and cascade if $v_j \in I_1$ or $C_1 \leq w(R)$.
Thus, the approximation ratio is at most
\begin{eqnarray}
  & & p'_1 \cdot \left[ p_j \cdot \frac{w(R') + w_1 + w_j}{w_1 + w_j} + (1-p_j) \cdot \frac{w(R') + w_1 + w_j}{w(R)} \right] \notag\\
  & & + (1 - p'_1) \cdot \left[ p_j \cdot \frac{w(R') + w_1}{C_1} + (1 - p_j) \cdot \frac{w(R') + w_1}{\min(C_1, w(R))} \right] \notag \\
  & \leq & p'_1 \cdot \left[ p_j \cdot \left(1 + \frac{w(R')}{w_1+w_j}\right) + (1-p_j) \cdot \left( 1 + \frac{1}{z} \right) \right] \notag \\
  & & + (1 - p'_1) \cdot \left[ p_j \cdot \left( 1 + \frac{w(R')}{w_1} \right) + (1 - p_j) \cdot \frac{w(R') + w_1}{\min(w_1, w(R))} \right] =: \phi_1, \label{eq:bound1p2}
\end{eqnarray}
where $z = w(R)/w_1$.

Now consider the case when the algorithm queries~$I_j$ before~$I_1$.
With probability~$p_j$, the optimum strategy is to query~$I_1$ and cascade, and the algorithm queries all intervals; also, if $v_1 \in I_j$, then $C_1 \geq w_1 + w_j$.
With probability $1-p_j$, the algorithm does not query~$I_1$; if $v_1 \in I_j$, then the optimum is to query~$R$ only, as $w(R') \leq w_1/3$ (as discussed in the previous paragraph), otherwise the optimum is the minimum between querying~$I_1$ and cascading or querying~$R$ only.
Therefore, the approximation ratio is at most
\begin{align}
 & && p_j \cdot \left[ (1-p'_1) \cdot \frac{w(R') + w_j + w_1}{C_1} + p'_1 \cdot \frac{w(R') + w_j + w_1}{w_1 + w_j} \right] \notag \\
 & && + (1-p_j) \cdot \left[p'_1 \cdot \frac{w(R') + w_j}{w(R)} + (1-p'_1) \cdot \frac{w(R') + w_j}{\min(C_1, w(R))} \right] \notag \\
 & \leq && p_j \cdot \left[ (1-p'_1) (1 + z) + p'_1 \cdot \left(1 + \frac{w(R')}{w_1+w_j}\right) \right] + (1-p_j) \cdot \left[ p'_1 + (1-p'_1) \cdot \frac{w(R)}{\min(w_1, w(R))} \right] =: \phi_j. \label{eq:bound2p2}
\end{align}

We need an upper bound for $\min(\phi_1, \phi_j)$.
We have two cases, depending on whether $w_1 \leq w(R)$.

\emph{Case 3.1}, $w_1 \leq w(R)$.
Note that $z \in [1, 4/3)$, where $z = w(R)/w_1$.
Also remember that $w(R') \leq w_j/3$ and $w(R') < w_1/3$, so $w(R') < (w_1 + w_j)/6$.
In this case, Equation~(\ref{eq:bound1p2}) is at most
\begin{eqnarray}
  &  & p'_1 \cdot \left[ p_j \cdot \frac{7}{6} + (1-p_j) \cdot \left( 1 + \frac{1}{z} \right) \right] + (1 - p'_1) \cdot \left[ p_j \cdot \frac{4}{3} + (1 - p_j) \cdot \frac{4}{3} \right] \notag \\
  & = & p_j \cdot p'_1 \cdot \left(\frac{1}{6} - \frac{1}{z}\right) + p'_1 \cdot \left( \frac{1}{z} - \frac{1}{3}\right) + \frac{4}{3}, \label{eq:bound1zg1}
\end{eqnarray}
which is a function of~$p_j$ that is non-increasing for any fixed value of $p'_1, z$, since $z \in [1, 4/3)$ and $p'_1 \geq 0$.
On its turn, Equation~(\ref{eq:bound2p2}) is at most (recall that $w(R)/w_1=z$)
\begin{eqnarray}
  & & p_j \cdot \left[ (1-p'_1) \cdot (1 + z) + p'_1 \cdot \frac{7}{6} \right] + (1-p_j) \cdot \left[ p'_1 + (1-p'_1) \cdot \frac{w(R)}{w_1} \right] \notag \\
  & = & p_j \cdot \left( 1 - \frac{5p'_1}{6} \right) + z \cdot (1-p'_1) + p'_1, \label{eq:bound2zg1}
\end{eqnarray}
 which is a function of~$p_j$ that is increasing for any fixed value of $p'_1, z$, since $p'_1 \in [0, 1]$.
Since one function is non-increasing and the other is increasing, and both are linear functions of~$p_j$, their minimum is maximized when they are equal.
Equating them, we obtain
\begin{displaymath}
  p_j = \frac{(1-p'_1)(4z-3z^2)+3p'_1}{3(z(1-p'_1)+p'_1)},
\end{displaymath}
and $p_j \in [0, 1]$ because $p'_1 \in [0, 1]$ and $z \in [1, 4/3)$.
Equation~(\ref{eq:bound2zg1}) then becomes
\begin{displaymath}
  \frac{((1-p'_1)(4z-3z^2)+3p'_1)(6-5p'_1)}{18(z(1-p'_1)+p'_1)} + z(1-p'_1)+p'_1 =: \hat{\phi}(p'_1, z),
\end{displaymath}
whose partial derivative on~$z$ is
\begin{displaymath}
  \frac{\partial}{\partial z} \hat{\phi}(p'_1, z) = \frac{p'_1(1-p'_1)(p'_1 + 6p'_1(2-z)+6-3z^2(1-p'_1))}{18(z(1-p'_1)+p'_1)^2} \geq 0,
\end{displaymath}
because $p'_1 \in [0, 1]$ and $z \in [1, 4/3)$.
Thus, for any fixed~$p'_1$, it holds that $\hat{\phi}$~is a non-decreasing function of~$z$, so it attains its maximum at $z$ tending to $4/3$.
Therefore,
\begin{displaymath}
  \min(\phi_1, \phi_j) < \hat{\phi}(p'_1, 4/3) = \frac{32+2p'_1-13{p'_1}^2}{24-6p'_1},
\end{displaymath}
which is maximized when $p'_1 = 4 - 2\sqrt{42/13}$ and the maximum is $17 - 2\sqrt{182/3} < 1.4223$.

\emph{Case 3.2}, $w_1 > w(R)$.
Note that $z \in (0, 1)$, where $z = w(R)/w_1$.
Recalling that $w(R') \leq w(R)/4$ and $w(R') \leq w_j/3$, we have $w(R') < w_1/4$, $w(R') < (w_1+w_j)/7$, and  $(w(R')+w_1)/\min(w_1,w(R))=(w(R')+w_1)/w(R)\le 1/z+1/4$.
In this case, Equation~(\ref{eq:bound1p2}) is at most
\begin{eqnarray}
  & & p'_1 \cdot \left[ p_j \cdot \frac{8}{7} + (1-p_j) \cdot \left( 1 + \frac{1}{z} \right) \right] + (1 - p'_1) \cdot \left[ p_j \cdot \frac{5}{4} + (1 - p_j) \cdot \left( \frac{1}{z} + \frac{1}{4}\right) \right] \notag\\
  & = & \frac{1}{z} \cdot (1-p_j) + p_j \cdot \left(1 - p'_1 \cdot \frac{6}{7} \right) + p'_1 \cdot \frac{3}{4} + \frac{1}{4}.\label{eq:bound2zl10}
\end{eqnarray}
which is a function of~$z$ that is non-increasing for any fixed value of $p'_1, p_j$, since $p_j \leq 1$.
On its turn, Equation~(\ref{eq:bound2p2}) is at most (using $w(R)/\min(w_1,w(R))=1$)
\begin{eqnarray}
  & & p_j \cdot \left[ (1-p'_1) \cdot (1 + z) + p'_1 \cdot \frac{8}{7} \right] + (1-p_j) \notag \\
  & = & z \cdot p_j \cdot (1-p'_1) + p_j \cdot p'_1 \cdot \frac{1}{7} + 1. \label{eq:bound2zl1}
\end{eqnarray}

If $p_j = 0$ then this is optimal, and when $p'_1 = 1$ it can be bounded by $8/7$.
Otherwise we have a function of~$z$ that is increasing for any fixed value of $p'_1, p_j$, since $p'_1 \in [0, 1)$ and $p_j \in (0, 1]$.
Since one function is non-increasing and the other is increasing, their minimum is maximized when they are equal, which happens when
\begin{displaymath}
  p_j = \frac{4 - 3z(1-p'_1)}{4(1-z(1-z)(1-p'_1))} \in [0, 1],
\end{displaymath}
and Equation~(\ref{eq:bound2zl1}) then becomes
\begin{displaymath}
  \frac{(4-3z(1-p'_1))(p'_1+7z(1-p'_1))}{28(1-z(1-z)(1-p'_1))} + 1 =: \tilde{\phi}(p'_1, z).
\end{displaymath}
Note that the denominator of the first summand in $\tilde{\phi}(p'_1, z)$ is minimized w.r.t.\ $z$ at $z=1/2$, hence is at most $21+7p'_1$, and the numerator is non-negative for $z \in (0, 1)$.
Taking factors $3(1-p'_1)$ and $7(1-p'_1)$ out of the first and second brackets (resp.) of the numerator and denoting $a=\frac{4}{3(1-p'_1)}$, $b=\frac{p'_1}{7(1-p'_1)}$, we have:
\[
\tilde{\phi}(p'_1, z)\le (a-z)(b+z)\cdot \frac{3(1-p'_1)^2}{3+p'_1} + 1\ .
\]
For fixed $p'_1$, the right-hand side of this inequality is a parabola with negative coefficient of $z^2$, hence it is maximized when $z=(a-b)/2$, implying
\begin{eqnarray*}
\tilde{\phi}(p'_1, z)\le \left(\frac{a+b}{2}\right)^2\cdot \frac{3(1-p'_1)^2}{3+p'_1} + 1&= &\left(\frac{28+3p'_1}{42(1-p'_1)}\right)^2\cdot \frac{3(1-p'_1)^2}{3+p'_1} + 1=\frac{(28+3p'_1)^2}{588(3+p'_1)}+1\\
&=&\frac{243}{196}+\frac{3p'_1}{196}+\frac{361}{588(p'_1+3)}.
\end{eqnarray*}
The derivative of this last function on $p'_1$ is $\frac{1}{588} \left( 9 - \frac{361}{(p'_1+3)^2} \right)$, which is negative for $p'_1 \in [0, 1)$.
So we have a decreasing function of $p'_1$, which is thus maximized at $p'_1=0$, giving us $\tilde{\phi}(p'_1, z)\le 13/9=1.4\bar 4$.
\qed
\end{proof}

The approximation factor can be improved slightly by choosing a different parameter $\delta \in (0, 1/3]$ such that $\alpha \in [\delta, 1-\delta]$.
Note that it does not work for $\delta \in (1/3, 1/2)$, since in this case we cannot build~$G'$ with $w(R)/2 \leq w(G') \leq (1-\delta) \cdot w(R)$ simply by assuming that every interval in~$R$ has cost less than $(1-\delta) \cdot w(R)$.
Using the analysis of Case~1, this implies that a lower bound on the approximation that can be achieved by this strategy is $1 + \sqrt{7}/6 > 1.4409$, even if not restricted by Cases~2 and~3.

\section{Further Questions}
\label{sec:future}

It would be interesting to find out whether it is possible to implement our dynamic programming algorithm for the sorting problem more efficiently, e.g., in $\Oh(n^4)$ time.
We could also try to extend our approach for sorting so as to handle a dynamic setting, e.g.\ as in~\cite{Busto:2019:MIP:3310435.3310582}, where some intervals can be inserted/deleted from the initial set; updating the dynamic program should be faster than building it again from scratch. 

The main question that we leave open regards the complexity of finding the best decision tree for the minimum problem.
It has proved non-trivial to obtain either a negative or positive result.
It could also be promising to look for better approximation algorithms for particular cases, e.g. if query costs and probability distributions are uniform.

If it is NP-hard to find the best decision tree for the minimum problem, then the same applies to the stochastic versions of the problem of finding the $k$-th smallest value (the generalized median problem)~\cite{feder03medianqueries,kahan91queries} and the minimum spanning tree problem with uncertainty on edge weights~\cite{erlebach14mstverification,erlebach08steiner_uncertainty,megow17mst}.
This holds because the minimum problem is a particular case of both: for the $k$-th smallest value problem it is the case with $k=1$; for the minimum spaning tree problem, if we have two vertices connected by multiple edges, then we want to identify the edge of minimum weight.
It would be interesting to find out whether it is possible to devise polynomial-time algorithms for those problems with better approximation guarantees than the respective best results for the adaptive online problem with adversarial inputs, as we did for the minimum problem.
We believe those problems would require different techniques to the ones we developed here, as they are more intricate.

\paragraph{Acknowledgements}
We would like to thank the reviewers for suggesting many improvements on the original draft.
In particular, one of the reviewers pointed the simplification of the recurrence in Section~\ref{sec:sortingdp}, which led to an improved time complexity regarding the preliminary conference paper.

\bibliographystyle{splncs04}
\bibliography{queries}
\end{document}